\documentclass[11pt, reqno]{amsart}

\usepackage{setspace}
\usepackage{fancyhdr,lipsum}
\usepackage{latexsym,graphics,enumerate}
\usepackage{amsmath,amsfonts,amsthm,amssymb}
\usepackage{mathrsfs}
\usepackage{ dsfont }
\usepackage{cleveref}

\usepackage{endnotes}

\let\footnote=\endnote

\newtheorem{theorem}{Theorem}[section]
\newtheorem{lemma}[theorem]{Lemma}
\newtheorem{prop}[theorem]{Proposition}
\newtheorem{cor}[theorem]{Corollary}
\theoremstyle{definition}
\newtheorem{definition}[theorem]{Definition}

\theoremstyle{remark}
\newtheorem{remark}[theorem]{Remark}

\newcommand{\Lp}[2]{\left\Vert \, #1 \, \right\Vert_{#2}}

\newcommand{\llp}[1]{ \Vert \, #1 \, \Vert }
\newcommand{\rr}{\mathbb{R}}
\newcommand{\cc}{\mathbb{C}}
\newcommand{\nn}{\mathbb{N}}

\newcommand{\bd}{\partial}
\newcommand{\lapl}{\Delta}
\newcommand{\inprod}[2]{\langle #1, #2 \rangle}

\newcommand{\grad}{\nabla}
\newcommand{\vect}[1]{\mathbf{ #1 }}

\newcommand{\Tr}{\operatorname{Tr}}
\newcommand{\ch}{\operatorname{ch}}
\newcommand{\sh}{\operatorname{sh}}
\newcommand{\D}{\text{d}}

\begin{document}

\title[Dynamics of Large Boson Systems with Attractive Interaction]{Dynamics of Large Boson Systems with Attractive Interaction and A Derivation of the Cubic Focusing NLS in $\rr^3$}


\author[Jacky Chong]{Jacky Chong\\
\tiny Department of Mathematics\\
\tiny The University of Texas, Austin, USA}
\address{Department of Mathematics, The University of Texas at
    Austin}
\curraddr{}
\email{jwchong@math.utexas.edu}
\thanks{}

\address{}
\curraddr{}
\email{}
\thanks{}




\dedicatory{}

\begin{abstract}
We consider a system of $N$ bosons where the particles
 experience a short range two-body interaction given by
 $N^{-1}v_N(x)=N^{3\beta-1}v(N^\beta x)$ where $v \in C^\infty_c(\rr^3)$,
 without a definite sign on $v$. 
We extend the results of M. Grillakis and M. Machedon, Comm. Math. Phys., \textbf{324}, 601(2013) and E. Kuz,  Differ.
 Integral Equ., \textbf{137}, 1613(2015) regarding the second-order correction to the mean-field evolution of systems with
  repulsive interaction to systems with attractive interaction for $0<\beta<\frac{1}{2}$. Our extension allows for a more general
   set of initial data which includes  coherent states. Inspired by the works of X. Chen and J. Holmer, Arch. Ration. Mech. Anal., 
   \textbf{221}, 631(2016) and 	Int. Math. Res. Not., \textbf{2017}, 4173(2017), and P. T. Nam and M. Napi\'orkowski,  Adv. Theor. 
   Math. Phys., \textbf{21}, 683(2017), we also provide both a
 derivation of the focusing nonlinear Schr\" odinger equation (NLS) in $3$D from the many-body system and its rate of convergence toward mean-field  for $0<\beta<\frac{1}{3}$. In particular, 
we give two derivations of the focusing NLS, one based on the $N$-norm approximation
given in the work of Nam and  Napi\'orkowski and the other via a method introduced in P. Pickl, J. Stat. Phys., \textbf{140}, 76(2010).  The techniques
used in this article are standard in the literature of dispersive PDEs. Nevertheless, the derivation of the focusing NLS 
 had only previously been studied for the 1D \& 2D cases and conditionally answered for the 3D case for $0<\beta<\frac{1}{6}$. 
\end{abstract}

\maketitle

\section{Introduction}
Bose-Einstein condensation is a physical phenomenon that
 occurs when a dilute gas\footnote{Dilute means the average particle separation distance
 is much bigger than the scattering length.}  of indistinguishable 
integer-spin particles\footnote{In relativistic quantum mechanics,
bosons are classified, by the Spin-Statistic theorem, to be particles with integer
 intrinsic spin. However, in this paper, we work in the realm of non-relativistic quantum physics where the bosonic property of particles is captured by the symmetry 
of the wave function.}   undergoes extreme cooling. Under this extreme condition, the gas of particles experiences  a phase
 transition where a macroscopic fraction of the particles coalesces into a single quantum state. 
 
 Historically, Bose-Einstein condensate (BEC) was predicted by Albert Einstein in the 
1920s\footnote{Einstein considered the non-interacting case.} 
long before its first realization in atomic gases by a series of experiments conducted on
 vapors of alkali metals in 1995. On June of 1995, for the first time, the JILA group led by
  Eric Cornell and Carl Wieman at the University of Colorado at Boulder NIST-JILA lab was
able to achieve a condensation limit in a gas of rubidium, ${}^{87}\text{Rb}$, inside a 
magnetic trap by using a combination of laser cooling and evaporative
 cooling techniques to lower the temperature of the substance to a mere $20 \text{ nK}$ (nano-kelvin). 
Shortly after the publication of the results of the JILA group, a group at MIT led by Wolfgang Ketterle 
was able to exhibit BEC using sodium, ${}^{23}\text{Na}$, but with many times more atoms
 than the experiment by the JILA group. In effect, Ketterle's group also 
demonstrated and measured many important properties of BEC.  Subsequently, the demonstrations 
by the two groups greatly increased both the experimental and theoretical activities in the
 field of large boson systems.  
 
 In recent years, many mathematics communities have made vigorous attempts at tackling 
the theory of many-body quantum mechanical systems to understand the evolution of condensates in the absolute zero temperature regime. One of the 
difficulties in modeling BEC is due to the size of the system. Since a system of $N$
interacting bosons is modeled by a symmetric wave function of $3N+1$ variables, 
the studies of the evolution of the wave function becomes impractical when $N$ is large, say
$N \sim 10^3$.\footnote{Many-body quantum systems are well studied in physics. In particular,
the general method for studying a large particle system is via a quantum statistical
 description using density matrices.}   
 Thus, it is favorable to find an effective description for the dynamics of the
large interacting boson system in a lower dimensional space. Informally, we would like to perform dimension reduction to reduce the original 
linear PDE of $3N+1$ variables to a nonlinear PDE with lesser variables, say 3+1, that captures the effective dynamics of the system. This desire leads to the studies of mean-field
approximation to the evolution of large particle systems. 

Despite the simplicity of the idea of trying to find an effective description for the dynamics of a large particle
 system\footnote{Equilibrium mean-field theory 
is well studied in any introductory course on statistical mechanics. For the non-equilibrium case, we refer the 
interested reader to the survey by F. Golse \cite{Gol}. }, a rigorous justification for the
 the  mean-field approximation is rather involved. In particular, the problem of finding an effective 
description for the evolution of BEC was only first studied systematically in a series of 
papers by Erd\"o and Yau, Elgart, Erd\"os, Schlein and Yau, and  Erd\"os, Schlein and Yau, \cite{EY, EESY, 
ESY, ESY2, ESY3, ESY4, ESY5}. Using the formalism of quantum BBGKY hierarchy, they were able to 
extract the mean-field limit as the number of particles tends to infinity and show that the limit satisfies
the defocusing cubic NLS.  Furthermore, this series of works drew the attention of the PDE community. 
Due to the complexity of the historical development of the studies of infinite hierarchies from the point of
view of dispersive PDEs, we refer the interested reader to a list of articles \cite{KM, KSS, CP, CP2, CX, CX2, CHHOL, CHHOL2, 
CHHOL3, CHHOL4, CHHOL5, GSS, Soh, Soh2} for a more in-depth 
view of the subject; the list is not intended to be a comprehensive collection of the available literature.

Let us briefly discuss the mathematical setting for our problem.  Consider an $N$-body boson system in $\rr^3$ whose dynamics
is governed by the $N$-body linear Schr\" odinger equation
\begin{align}
\frac{1}{i}\frac{\bd}{\bd t}\Psi_N = H_N\Psi_N= \left( \sum^N_{j=1} \lapl_{x_j} -\frac{1}{N} \sum_{i<j} v_N(x_i-x_j)\right)\Psi_N 
\end{align}
with factorized initial datum\footnote{It should be noted that the ground state of the system, in general, cannot be approximated by a factorized state. However, it is
expected to be factorized in the large particle limit. Aside from studying dynamics around the ground state, there is also an interest in studying the dynamical formation of correlations starting with unentangled states. }, i.e. $\Psi_N(0, x_1, \ldots, x_N) = \prod^N_{j=1}\phi_0(x_j) =\phi^{\otimes N}$. This setting provides us with an appropriate model for studying the evolution of BEC\footnote{Cf. Chapter 1.3 and 7 of \cite{Lieb}.}.

Formally, we say an $N$-body boson system exhibits  the \emph{complete BEC}  property provided the \emph{one-particle
marginal density operator}, $\gamma^{(1)}_N$, factorizes in trace norm as $N\rightarrow \infty$, i.e. 
\begin{align}
\Tr\big|\gamma_N^{(1)}-|\phi\rangle\langle \phi|\big|\rightarrow 0 \ \ \text{ as } \ \ N \rightarrow \infty
\end{align}
for some $\phi$. Let us note the kernel of $\gamma^{(1)}_N$ is given by
\begin{align}
\gamma_N^{(1)}(x, x') = \int d\vect{x}\ \Psi_N^\ast(x, \vect{x})\Psi_N(x', \vect{x}) \ \ \ x, x'\in \rr^3\ \text{ and }  \ \vect{x} \in \rr^{3(N-1)}.
\end{align}
Indeed, using this definition, one can show that the evolution of BEC can be effectively approximated by a one-body mean-field dynamics; see \cite{EY, EESY, ESY, ESY2, ESY4, ESY5}.  

A natural question one could ask is whether the above statement holds true in state space. More specifically, if we start with a factorized initial state
then is it true that under time evolution the many-body wave function can be approximated by
\begin{align}
\Psi_N(t, x_1, \ldots, x_N) \sim e^{i\chi(t)}\phi^{\otimes N}=e^{i\chi(t)}\prod^N_{j=1}\phi(t, x_j) 
\end{align}
in $L^2(\rr^{3N})$ as $N\rightarrow \infty$ for some phase $\chi(t)$? Unfortunately, the answer is negative. However, in recent years, many have considered
 initial states of the form
\begin{align}\label{groundstate}
\Psi_N= \sum^N_{n=0} \phi^{\otimes (N-n)}\otimes_s \psi_n
\end{align}
where $(\psi_n)_{n=1}^\infty$ is a family of functions with increasing number of variables that models the behavior of the wave function outside the condensate $\phi$.
The form of \eqref{groundstate} is motivated by properties of the ground state of the many-body system; see \cite{LNSS}. It has been shown in \cite{LNS, NaNa, NaNa3} 
that for systems with repulsive interaction and $0<\beta<\frac{1}{2}$, the
evolution $\Psi_N(t)= e^{itH_N}\Psi_N$ satisfies the norm approximation
\begin{align}
\lim_{N\rightarrow \infty}\Lp{\Psi_N(t)-\sum^N_{n=0} \phi(t)^{\otimes (N-n)}\otimes_s \psi_n(t)}{L^2(\rr^{3N})}=0
\end{align}
where $\phi$ is a solution to some Hartree-type equation and $(\psi_n(t))$ is generated by a quadratic Bogoliubov Hamiltonian. In fact, by imposing additional structures on both the initial data and the quadratic Bogoliubov Hamiltonian, it was shown in \cite{BNNS} that the result also holds true for $0<\beta<1$.  The approach in this article is similar in spirit
to the above norm approximation. However, we consider the problem in a state space that allows an indefinite (varying) number of particles, which we called  the Fock space and 
obtain a Fock space norm approximation. 

\subsection*{Acknowledgements} The author would like to thank his two Ph.D. advisors Professor M. Grillakis and Professor M. Machedon for many hours of useful discussion.  
They are both inspiration role models for the author both in research and in life. The author would also like to thank the referee for his/her careful review of the manuscript and critical suggestions for improving the overall
quality and presentation of the paper. 

\section{Earlier Results and Main Statements}
\subsection{Background and Earlier Results}
This section provides a brief overview of the results obtained in 
\cite{GM1, GM2, GM3, Kuz} along with some background materials for the convenience 
of the reader. 

Let us introduce the mathematical setting for our work. The one-particle base space, denote by $\mathfrak{h}:= L^2(\rr^3, dx )$, is a complex
separable Hilbert space endowed with the inner product $\inprod{\cdot}{\cdot}_{\mathfrak{h}}$ that is linear
in the second variable and conjugate linear (or anti-linear) in the first variable
\footnote{This is the physicists' inner product.}.

We define the \emph{bosonic Fock space over $\mathfrak{h}$} to be the closure of
\begin{align}
\mathcal{F}_s(\mathfrak{h}) = \mathcal{F}_s := \cc\oplus \bigoplus_{n= 1}^\infty \operatorname{Sym}(\mathfrak{h}^{\otimes n})
\end{align}
with respect to the norm induced by the \emph{Fock inner product}
\begin{align}
\inprod{\varphi}{\psi}_\mathcal{F} = \bar{\varphi_0}\psi_0 + \sum^\infty_{n=1} 
\inprod{\varphi_n}{\psi_n}_{\mathfrak{h}^{\otimes n}}
\end{align}
where $\varphi=(\varphi_0, \varphi_1, \ldots), \psi=(\psi_0, \psi_1, \ldots)
 \in \mathcal{F}_s(\mathfrak{h})$. For convenience, we shall
 refer to $\mathcal{F}_s$ simply as the Fock space and drop the subscript henceforth. The \textit{vacuum}, denote by
 $\Omega$, is defined to be the Fock vector $(1, 0, 0, \ldots ) \in \mathcal{F}$. 
 
 For every field $\phi \in \mathfrak{h}$, we define the associated
 \textit{creation} and \textit{annihilation} operators on
 $\mathcal{F}$, denoted respectively by
 $a^\ast(\phi)$ and $a(\bar\phi)$, as follow
 \begin{subequations}
\begin{align}
(a^\ast(\phi)\psi)_n(x_1, \ldots, x_n) :=&  \frac{1}{\sqrt{n}} \sum^n_{j=1} \phi(x_j)
 \psi_{n-1}(x_1, \ldots, \widehat x_j, \ldots, x_n) \label{creation}\\
(a(\bar\phi)\psi)_n(x_1, \ldots, x_n):=& \sqrt{n+1} \int d x\ \bar\phi(x) 
\psi_{n+1}(x, x_1, \ldots, x_n) \label{annihilation}
\end{align}
\end{subequations}
with the property that $a(\phi)\Omega = 0$.
We can also define the corresponding creation and annihilation distribution-valued 
operators associated to \eqref{creation} and \eqref{annihilation}, denote by $a^\ast_x$ and $a_x,$ as
follow
\begin{subequations}
\begin{align}
(a^\ast_x\psi)_n := &\ \frac{1}{\sqrt{n}} \sum^n_{j=1} \delta(x-x_j) \psi_{n-1}(x_1, \ldots, \widehat x_j, \ldots, x_n)\\
(a_x\psi)_n :=&\ \sqrt{n+1}\psi_{n+1}(x, x_1, \ldots, x_n).
\end{align}
\end{subequations}
In short, we have the relations
\begin{align}
a^\ast(\phi) = \int d x\ \{\phi(x) a^\ast_x\} \ \ \text{ and } \ \ a(\bar\phi) = \int d x\ \{\bar\phi(x) a_x\}.
\end{align}
Let us note that the creation and annihilation operators\footnote{It should be warned that we follow the convention of \cite{GM2} and define our 
annihilation operator to be a  linear map as opposed to the conventional definition of anti-linear. This definition is also consistent 
with the view that $a_x$ is a distribution-valued operator, since $a_x\psi$ acts linearly on $\mathfrak{h}$. 
} $a(\bar\phi)$ and $a^\ast(\phi)$
 associated to the field $\phi$ are unbounded, densely defined,
closed operators. Moreover, one can easily verify, formally, $(a^\ast_x, a_x)$ satisfy the canonical commutation relations (CCR):
$[a_x, a_y^\ast] = \delta(x-y)$, $[a_x, a_y] = [a_x^\ast, a_y^\ast] = 0$\footnote{The reader should note for any
$f, g \in \mathfrak{h}$ the CCR for $a^\ast(f)$ and $a(g)$ are not well defined since there are domain issues that need to be resolved for the given unbounded operators. 
For an exotic example of an ill-defined commutator of unbounded operators, we refer the reader to Chapter VIII.5 of \cite{ReSim}. }, and 
the \emph{number operator} defined by
\begin{align}
 \mathcal{N}:= \int d x\ a_x^\ast a_x
\end{align}
is a diagonal operator on $\mathcal{F}$ that counts the number of particles in each and every sector.

As mentioned in the introduction, we are interested in studying the 
time evolution of uncorrelated states or states around the ground state of an $N$-body system in the Fock space setting. 
To this end, it is convenient to define a special class of initial
 data, the coherent states (or, more generally, squeezed states which we will define shortly),
and the Fock Hamiltonian.  

For each $\phi \in \mathfrak{h}$, we associate the corresponding unique closure of the
 operator
\begin{align}
 \mathcal{A}(\phi) = a(\bar\phi) - a^\ast(\phi)
\end{align}
then the \emph{Weyl operator}\footnote{To avoid the unfavorable technicality associated with the unbounded natural of our creation and annihilation operators, one often choose to work with the corresponding Weyl algebra, the $C^\ast$-algebra generated by the exponential of $\mathcal{A}(\phi)$ where $\phi \in \mathfrak{h}$ (cf. Chapter 9 of \cite{DezGer} and Chapter 5.2 of \cite{BraRob}).} is defined by
\begin{align}
 e^{-\sqrt{N}\mathcal{A}(\phi)}.
\end{align}
Let us note that the operator $\mathcal{A}(\phi)$ is a skew-Hermitian unbounded operator which means
the corresponding Weyl operator is unitary. The \emph{coherent state} associated to the field $\phi$
is defined by 
\begin{align}
 \psi(\phi):=e^{-\sqrt{N}\mathcal{A}(\phi)}\Omega.
\end{align}
Using the Baker-Campbell Hausdorff formula, one can show that 
\begin{align}
 e^{-\sqrt{N}\mathcal{A}(\phi)}\Omega = \left(\ldots, c_n \phi^{\otimes n}, \ldots\right) \ \ 
\text{ where } \ \ c_n = \left( e^{-N\llp{\phi}^2_\mathfrak{h}}N^n/n!\right).
\end{align}
 For a fixed $N \in \nn$, we defined the \emph{Fock Hamiltonian} (associated 
to $N$), denoted by $\mathcal{H}$, to be the diagonal operator on the Fock space 
given by 
\begin{align}
 (\mathcal{H}\psi)_n = \left(\sum^n_{j=1}\lapl_{x_j}
-\frac{1}{N}\sum_{i<j}^nv_N(x_i-x_j)\right)\psi_n=:H_{N, n}\psi_n
\end{align}
where $v_N(x) = N^{3\beta}v(N^{\beta}x)$. Rewrite $\mathcal{H}$ using creation and annihilation operators, we get that \footnote{Here, we also adopt the convention of \cite{GM2} and
define $\mathcal{H}$ so that $\mathcal{H}\le 0$.} 
\begin{subequations}
\begin{align}
 \mathcal{H} :=&\ \mathcal{H}_1-\frac{1}{N}\mathcal{V}, \label{Fock-Ham}\\
 \mathcal{H}_1 :=&\ \int  dxdy\ \{\lapl_x \delta(x-y) a^\ast_x a_y\}, \ \ \text{ and }\\
 \mathcal{V} :=&\ \frac{1}{2}\int dxdy\ \{v_N(x-y)a^\ast_x a^\ast_y a_xa_y\}.
\end{align}
\end{subequations}
In light of \eqref{Fock-Ham}, we are interested in the
solution to the following Cauchy problem
in Fock space
\begin{align}
 \frac{1}{i}\frac{\bd}{\bd t} \psi = \mathcal{H}\psi \ \ \text{ with initial datum } \ \ \psi_0 = e^{-\sqrt{N}
\mathcal{A}(\phi_0)}\Omega
\end{align}
which is given by
\begin{align}\label{exact}
 \psi_{\operatorname{exact}}= e^{it\mathcal{H}}e^{-\sqrt{N}\mathcal{A}(\phi_0)}\Omega.
\end{align}

An important fact to note about the Fock Hamiltonian is its action on the $N$th sector 
of the Fock space. There, the Fock Hamiltonian acts as a mean-field Hamiltonian for the $N$-particle system, that is
\begin{align}
 (\mathcal{H}\psi)_N = \left(\sum^N_{j=1}\lapl_{x_j}
-\frac{1}{N}\sum_{i<j}^Nv_N(x_i-x_j)\right)\psi_N =: H_{\operatorname{mf}}\psi_N.
\end{align}
Since the $N$th coefficient $c_N$ could be approximated by $N^{-\frac{1}{4}}$ using Stirling's 
formula and the coherent state is a simple $N$-tensor of $\phi$ in the $N$ sector, then, heuristically, 
we see that by understanding the evolution of the coherent state we could also
 understand the mean-field evolution of the $N$-particle factorized state. However, we do not take this point of view in our studies since there is a better alternative way to consider
 the $N$-norm approximation which we have already mentioned in the introduction. Nevertheless, we want to emphasize the point that there are differences between the $N$-norm approximation and 
 the Fock norm approximation; one obvious difference is due to the  weight $c_N$ ($\sim N^{-\frac{1}{4}}$ factor). 

Based on the earlier works of Hepp in \cite{Hepp} and Ginibre \& Velo in  \cite{GV, GV2}, Rodnianski and Schlein in \cite{RS} study the one-particle Fock marginal, which is defined as follows: for every $\psi \in \mathcal{F}_s$ 
the \emph{one-particle Fock marginal of} $\psi$, denoted by $\Gamma^{(1)}_\psi$, is a 
positive trace class integral operator on $\mathfrak{h}$ whose kernel is given by
\begin{align}
 \Gamma_{\psi}^{(1)}(x, y)= \frac{\inprod{\psi}{a^\ast_xa_y\psi}_\mathcal{F}}{\inprod{\psi}{\mathcal{N}\psi}_{\mathcal{F}}}.
\end{align} 
They were able to show that the one-particle Fock marginal with an initial coherent state converges to the Hartree dynamics in trace norm for the case $\beta =0$. Furthermore, they were also able to obtain a rate of convergence
\begin{align}
 \Tr\left|\Gamma^{(1)}_{N,t} - |\phi_t\rangle \langle \phi_t|\right|\lesssim \frac{e^{Kt}}{N}
\end{align}
 for some constant $K>0$ where $\Gamma^{(1)}_{N, t}$ denotes the one-particle 
Fock marginal density of $\psi_{\operatorname{exact}}$ and $\phi_t$ satisfies the Hartree equation. Later, 
Kuz in \cite{Kuz} improved the estimate substantially in time and obtain the estimate
\begin{align}\label{linear-bound}
 \Tr\left|\Gamma^{(1)}_{N,t} - |\phi_t\rangle \langle \phi_t|\right|\lesssim \frac{t}{N}.
\end{align}
Unlike the approach of Rodnianski and Schlein which uses the mean-field approximation of the form
\begin{align}\label{mfapprox}
\psi_{\text{mf}} = e^{-\sqrt{N}\mathcal{A}(\phi_t)}\Omega=e^{-\sqrt{N}\mathcal{A}(t)}\Omega,
\end{align}
 Kuz uses the method of second-order correction introduced in the works of Grillakis, Machedon and Margetis in  \cite{GMM, GMM2, GM1} to establish \eqref{linear-bound}, which relies
on tracking the exact dynamics of the evolution of the coherent state in Fock space. 

To track  the exact dynamics in Fock space, we need to introduce the \emph{pair excitation function}, $k(x, y)=k(y, x)$, and its corresponding quadratic operator $\mathcal{B}(k)$ defined by
\begin{align}
\mathcal{B}(k_t) = \mathcal{B}(t)= \int \D x\D y\ \{\bar k(t, x, y)a_x a_y - k(t, x, y)a^\ast_x a^\ast_y\}.
\end{align}
From the pair excitation, we concoct a new approximation scheme, which is a second-order correction\footnote{In the mathematical physics literature, $e^{\mathcal{B}}$ is called the infinite dimensional Segal-Shale-Weil Representation of the double cover of the group of symplectic matrices of integral operators. The elements of the corresponding $C^\ast$-algebra are called the Bogoliubov transformations (cf. Chapter 4 of \cite{Folland} and Chapter 11 of \cite{DezGer}).  } to the mean-field \eqref{mfapprox},   given by
\begin{align}\label{approx}
\psi_{\text{approx}} = e^{iN\chi (t)}e^{-\sqrt{N}\mathcal{A}(t)}e^{-\mathcal{B}(t)}\Omega
\end{align}
where $\chi(t)$ is some phase factor to be determined. With some appropriate choice of evolution equations for $\phi$ and $k$, we will later see that 
\eqref{approx} will indeed allows use to track the exact dynamics of the evolution of coherent states or states of the form $e^{-\sqrt{N}\mathcal{A}(\phi_0)}e^{-\mathcal{B}(k_0)}\Omega$, called \emph{squeezed states}. 

Incidentally, one could show via a Lie algebra isomorphism argument that the evolution of $k$ is best described in terms of some nonlinear evolution equations of the fields
\begin{subequations}
\begin{align}
\sh(k):=&\ k+\frac{1}{3!}k \circ \bar k \circ k + \frac{1}{5!}k \circ \bar k \circ k \circ\bar k \circ k+\ldots\\
\ch(k):=&\ \delta+\frac{1}{2!}\bar k \circ k +\frac{1}{4!}\bar k \circ k \circ \bar k \circ k+ \ldots
\end{align}
\end{subequations}
where $\circ$ denotes the composition of operators; see \cite{GM1, GM2, GM3, Na} for more details. Moreover, in \cite{GM2}, Grillakis and Machedon show, by using a specific 
coordinate, that the nonlinear equations for the pair excitation function could be express as a system of coupled linear equations in $\sh(2k)$ and $\ch(2k)$. Note, we also have 
the identity $\sh(2k)= 2\sh(k)\circ \ch(k)$ and $\ch(2k)-\delta = 2\overline{\sh(k)}\circ \sh(k)$. 
 
 Let us introduce some notation to help us compactly write out the evolution equations for $\phi$ and $k$. We write
 \begin{subequations}
\begin{align}
 g(t, x, y):=& -\lapl_x \delta(x-y) + (v_N\ast |\phi|^2)(t, x)\delta(x-y) \\
&+ v_N(x-y)\bar\phi(t, x) \phi(t, y) \nonumber\\
m(t, x, y):=& -v_N(x-y)\phi(t, x)\phi(t, y) \\ 
&\ (\text{we also write } m = v_N\phi\otimes\phi) \nonumber\\
V(t, x, y):=&\ [(v_N\ast |\phi|^2)(t, x)+(v_N\ast |\phi|^2)(t, y)]\delta(x-y)\\
&+ v_N(x-y)\bar\phi(t, x) \phi(t, y) + v_N(x-y)\phi(t, x) \bar\phi(t, y) \nonumber
\end{align}
\end{subequations}
and define the operators
\begin{subequations}
\begin{align}
\vect{S}(s):=&\  \frac{1}{i}\bd_ts +g_N^T\circ s+s\circ g_N \ \ (\text{Schr\" odinger-type operator})\\
=&\  \frac{1}{i}\bd_ts+\{-\lapl, s\}+V\circ s \ \ (\text{we also write } V\circ s=V(s))\nonumber\\
\vect{W}(p):=&\ \frac{1}{i}\bd_t p+ [g_N^T, p] \ \   (\text{Wigner-type operator})
\end{align}
\end{subequations}
then the desired evolution equations of $\phi$ and $k$ are given by
\begin{subequations}\label{uncoupled}
\begin{align}
&\frac{1}{i}\bd_t \phi -\lapl_x \phi +(v_N\ast |\phi|^2)\phi = 0 \ \ \ \ \  (\text{Hartree-type equation}) \label{eeq1}\\
&\vect{S}(\sh(2k))= m\circ \ch(2k)+\overline{\ch(2k)}\circ m \label{eeq2}\\
&\vect{W}(\ch(2k)) = m\circ\overline{\sh(2k)}-\sh(2k)\circ\overline{m} \label{eeq3}.
\end{align}
\end{subequations}
The system of equations \eqref{uncoupled} is referred to as the uncoupled system as opposed to the coupled system (time-dependent Hartree-Fock-Bogoliubov (HFB) system) studied in \cite{GM1, GM3, BBCFS} where the equation for the condensate are coupled with the pair excitation equations.
 
 Now, let us summarize the results in \cite{GM2, Kuz}, which built on earlier works of Grillakis, Machedon, and Margetis
 in \cite{GMM, GMM2}.
\begin{theorem}
Let $v \in C^\infty_c(\rr^3)$ and $v\geq 0$.
Assume $\phi$ and $k$ satisfy \eqref{uncoupled} with initial conditions $\phi(0, \cdot) = \phi_0 \in L^2(\rr^2)\cap W^{m, 1}(\rr^3)$ for
some sufficiently large $m$ and $ k(0, \cdot) = 0$. If $\psi_{\text{exact}}$ and $\psi_{\text{approx}}$ are defined by \eqref{exact} and \eqref{approx} respectively, then we have the following estimate
\begin{align}
\llp{\psi_{\text{exact}(t)}-\psi_{\text{approx}}(t)}_\mathcal{F} \lesssim  \frac{(1+t)\log^4(1+t)}{N^{(1-3\beta)/2}}
\end{align}
provided $0 < \beta<\frac{1}{3}$. Moreover, if $(\bd_t \sh(2k))(0, \cdot)$ is sufficiently regular, then for any $\epsilon>0$ and $j$ a positive integer, we have
\begin{align}
&\llp{\psi_{\text{exact}}(t)-\psi_{\text{approx}}(t)}_\mathcal{F} \nonumber\\
& \lesssim t^{\frac{j+3}{2}}\log^6 (1+t) \cdot
\begin{cases}
N^{-\frac{1}{2}+\beta(1+\epsilon)} &  \frac{1}{3} \leq \beta < \frac{2j}{(1-2\epsilon +4j)},\\
N^{\frac{-3+7\beta}{2}+(j-1)(-1+2\beta)} &  \frac{2j}{(1-2\epsilon +4j)} \leq \beta < \frac{1+2j}{3+4j}. 
\end{cases}
\end{align}
\end{theorem}

\begin{remark}
 It should be noted that the assumption $(\bd_t \sh(2k))(0, \cdot)$ must be sufficiently
regular imposes a restriction on the form of the initial condition; in particular, $k(0, \cdot)$ cannot be zero. 
Due to the restriction, Kuz could not choose the coherent state as the initial condition since 
$e^{-\sqrt{N}\mathcal{A}_0}
e^{-\mathcal{B}_0}\Omega$ is a coherent state if and only if $k(0, \cdot) = 0$.  Nevertheless, the condition allows for states close to the ground state. 
\end{remark}

\subsection{Main Statements}
The main purpose of this article is to extend the results in \cite{GM2, Kuz, Kuz2}
 to the case of arbitrary $v \in C^\infty_c$ with small $\dot H^{\frac{1}{2}}_x$ data for $\phi$  and get rid of the constraint of $(\bd_t \sh(2k))(0, \cdot)$  given in \cite{Kuz2}.   Let us state the main result of our work.
\begin{theorem}\label{main1}
Let $v \in C^\infty_c(\rr^3)$. Assume $\phi$ and $k$ satisfy \eqref{uncoupled} with initial conditions $\phi_0 \in L^2(\rr^3)\cap W^{m, 1}(\rr^3)$ and $\llp{\phi_0}_{L^2_x}=1$ 
for some $m$ sufficiently large and  $\dot H^{\frac{1}{2}}_x$-norm sufficiently small, depending on $v$, and $ k(0, \cdot) = 0$ (or more general  smooth data). If $\psi_{\text{exact}}$ and $\psi_{\text{approx}}$
 are defined by \eqref{exact} and \eqref{approx} respectively, then, for any $\varepsilon>0$, we have the Fock space estimate
\begin{align}
\llp{\psi_{\text{exact}}(t)-\psi_{\text{approx}}(t)}_\mathcal{F} \lesssim_\varepsilon  N^{-\frac{1}{2}+\beta(1+\varepsilon)}+P_N(t)
\end{align}
where $P_N(t)$ is some (quadratic) polynomial  with the property: on any fixed interval $[0, T]$, we have that
\begin{align}
|P_N(t)| \ll N^{-\frac{1}{2}+\beta(1+\varepsilon)}
\end{align}
on $[0, T]$ when $N$ is sufficiently large provided $0 < \beta<\frac{1}{2(1+\varepsilon)}$. If $v\ge 0$ then the assumption on the smallness of $\dot H^{\frac{1}{2}}_x$-norm can be dropped. 
\end{theorem}

\begin{remark}\label{beta<1}
It has already been remarked in \cite{Kuz2} that the range of $\beta$ in Theorem \ref{main1} is optimal for the uncoupled system \eqref{uncoupled}. The constraint comes from estimating $\mathcal{P}_1$ in Lemma
\ref{p1-lem}; see Appendix for more detail. However, it has been shown in \cite{GM3, GM4, Chong} that the range of $\beta$ can be extended to $0<\beta<1$ 
provided we consider 
the time-dependent HFB system and $v\ge 0$. An earlier result of Boccato, Cenatiempo, and Schlein in \cite{BCS}  have also shown that the Fock space norm approximation holds
for $0<\beta<1$; the authors assume an explicit form for the pair excitation $k$ and work with a specific class of initial data. The main difference between the two approaches is the fact that the former group writes down a nonlinear system of PDEs  for the pair excitation function and approach the problem from a dispersive PDE perspective. However, in
both cases, the bound on the Fock space error is either exponential or double-exponential in time. 

The case $\beta =1$ is physically interesting since it corresponds to the Gross-Pitaveskii scaling regime. However, it is also mathematically the most difficult case.
 In this situation, it is not clear whether we have Fock space norm approximation.  In particular, both Fock space estimates provided in \cite{BCS, Chong} breaks down precisely when $\beta =1$. 
 Nevertheless, there are important results of Bogoliubov theory applied to the studies of dynamics of interacting bosons in the Gross-Pitaveskii regime. We refer the reader to \cite{BdOS, BS}
 for a more in-depth coverage of the topic.
\end{remark}

\begin{remark}\label{soliton}
For the focusing NLS, we know that it is globally well-posed in $H^1(\rr^3)$ and its solution scatters if the initial data is below some threshold given 
by the nonlinear ground state of the NLS; otherwise, the solution blows up  (See \cite{DHR, DM2}). Moreover, for each fixed $N$, it is well known that \eqref{eeq1} with compactly 
supported potential $v$ is globally
well-posed for any initial datum in $H^1$  even if $v\le 0$; see Theorem 3.1 in \cite{GV3}. However, the uniform in $N$ global well-posedness of \eqref{eeq1} in $H^1$ 
is not clear. 

The assumption on the smallness of the $\dot H^\frac{1}{2}$ norm
 guarantees the uniform in $N$ global well-posedness of the family of Hartree-type equations \eqref{eeq1} and provides us with uniform in $N$ a-priori bounds. However, we do not quantify how small
the data has to be. Despite the fact that the family of solutions to \eqref{eeq1} converges to the solution of the focusing NLS as $N$ tends to infinity, we are not able to 
precisely define a threshold for the uniform in $N$ global well-posedness of the family of equations.  
\end{remark}

\begin{remark}
For $v\ge 0$, the analysis of $\phi$ is given in \cite{GM2}.  The improvement in Section 4 allows us to improved the results in \cite{GM2} and \cite{Kuz2}. In fact, the crucial
ingredient for everything in Section 4 to hold is the fact that we have the uniform in $N$ time-decay estimates for $\phi$, not the sign of $v$. 
\end{remark}

The second purpose of the article is to derive the focusing cubic NLS in $\rr^3$ from 
a many-body boson system as in \cite{CHHOL, CHHOL3, CHHOL2, NaNa, JP, NaNa2}. 
 For this purpose, we assume $v \leq 0$, i.e. the
 interaction is attractive. In this case, we have the following statement.

\begin{theorem}\label{DerFNLS}(Factorized Initial Condition)
Assume $v \in C^\infty_c(\rr^3)$ and $v\leq 0$. Suppose $\Psi_N(t, \vect{x})$ solves the initial value problem
\begin{align}
 \frac{1}{i}\bd_t\Psi_N(t, \vect{x}) = H_{\text{mf}} \Psi_N(t, \vect{x}), \ \ \ \Psi_N(0, \cdot) 
= \phi_0^{\otimes N}
\end{align}
where $\phi_0$ satisfies the same conditions as in Theorem \ref{main1}. Denote the one-particle density associated to $\Psi_N(t, x)$
 by $\gamma^{(1)}_{N, t}$. Then we have the estimate
\begin{align*}
  \Tr\left|\gamma_{N}^{(1)}(t, \cdot) - |\phi_t\rangle \langle \phi_t|\right| 
\lesssim C(t) N^{\delta} 
\end{align*}
for some $\delta<0$ where
\begin{align}
C(t)
=
\begin{cases}
1 & \text{ if } 0<\beta<\frac{1}{6}\\
(1+t)\log^4(1+t) & \text{ if } \frac{1}{6}\le \beta <\frac{1}{3}
\end{cases},
\end{align}
and $\phi_t$ solves the focusing NLS
\begin{align}\label{FNLS}
\frac{1}{i}\bd_t \phi -\lapl_x \phi +(\int v)|\phi|^2\phi = 0.
\end{align}
\end{theorem}

\begin{remark}
The reader should note that Theorem \ref{DerFNLS} only addresses the derivation of the focusing NLS for a system of weakly-interacting dense Bose gas since $\beta \in (0, \frac{1}{3})$. 
\end{remark}

\begin{remark}\label{soliton2}
In the case $0<\beta<\frac{1}{6}$, we prove Theorem \ref{DerFNLS} by applying Pickl's method which we introduce in Section 6.1. There we do not need to work with a family of Hartree-type equation \eqref{eeq1}. Instead, we work directly with 
\eqref{FNLS}. As a consequence, we are able to apply standard facts about focusing NLS (see Remark \ref{soliton}) which allows us to work with any $\phi_0 \in H^1$ below a ground state threshold (i.e. we can handle any initial data that does not give raise to soliton solutions). 
\end{remark}

\section{Estimates for the Solution to the Hartree Equation}
Let us consider the following family of Hartree-type equations
\begin{align}\label{htnls}
 &\frac{1}{i}\bd_t\phi -\lapl \phi +(v_N\ast |\phi|^2)\phi = 0\\
 &\phi(0, \cdot) = \phi_0  \ \text{with } \ \llp{\phi_0}_{L^2_x}=1\nonumber
\end{align}
where $v_N(x) = N^{3\beta}v(N^\beta x)$
for $0\leq \beta \leq 1$ and $v \in C^\infty_0(\rr^3)$ is radial but not necessary nonnegative.  In this section, we prove the uniform in $N$ global well-posedness
of the Hartree-type equation for small data and its corresponding decay estimates.  

\subsection{Uniform in $N$ Global Wellposedness}
In this subsection, we prove the uniform in $N$ global well-posedness of \eqref{htnls}  assuming some small Sobolev condition.  

Let us begin by adopting some 
standard notations in dispersive PDE theory. We write $A \lesssim B$ to denote there exists a constant $C>0$ such that $A\le CB$. Consider the functions $f(x)$ and $g(x, t)$, we write
\begin{align*}
\llp{f}_{L^r_x}= \left(\int_{\rr^d} dx\ |f(x)|^r\right)^{\frac{1}{r}}, \ \ \llp{g}_{L^{q}_tL^r_x} = \left(\int^\infty_{-\infty} dt\ \llp{g(\cdot, t)}_{L^r(dx)}^q\right)^{\frac{1}{q}}
\end{align*}
with the usual adjustment in the case of $q$ or $r$ equals $\infty$. We define the Fourier transform and the space-time Fourier transform by
\begin{align*}
\widehat f(\xi) = \int_{\rr^d} dx\ e^{-ix\cdot \xi} f(x), \ \ \widetilde g(\xi, \tau) = \int_{\rr^{d+1}} dxdt\ e^{-ix\cdot \xi-it\tau} g(x, t),
\end{align*}
and sometimes write $\mathcal{F}(f)(\xi) = \widehat f(\xi)$ and $\mathcal{F}(g)(\xi, \tau) = \widetilde g (\xi, \tau)$ (should be clear from the context).
We also define the homogeneous Sobolev norm by 
\begin{align*}
\llp{f}_{\dot H_x^{s}}=\llp{|\grad|^{s} f}_{L^2_x} =\llp{D^{s} f}_{L^2_x}:= \left(\int d\xi\ |\xi|^{2s}|\widehat{f}(\xi)|^2\right)^{\frac{1}{2}}.
\end{align*}
In the case of partial spatial derivatives, we use the standard notation $\bd^\alpha f$ where $\alpha \in \nn^{d}$. We use interpolation to define fractional partial derivatives. 
If time differentiation is involved, we will make the notation more specific by denoting with subscript, i.e. $\bd_t^j f$.  Moreover, the general $L^p$ fractional Sobolev space is defined through complex interpolation. 

 A pair of numbers $(q, r)$ is admissible provided $q, r \geq 2$ and 
$\frac{2}{q}+\frac{d}{r} = \frac{d}{2}$ (for simplicity, we specialize to the case of 3D admissible, i.e. $\frac{2}{q}+\frac{3}{r} = \frac{3}{2}$). 
 Then the Strichartz norm  and its dual norm are defined by
\begin{align*}
\llp{g}_{S^0} := \sup_{(q, r) \ \text{admissible}} \llp{g}_{L^{q}_tL^r_x}, \ \ \llp{g}_{N^0} := \inf_{(q, r) \ \text{admissible}} \llp{g}_{L^{q'}_tL^{r'}_x}
\end{align*} 
where $q', r'$ are the H\"older conjugates of $q, r$. In fact, in our work, the notation also means: let $u(x, y, t)$  be a function of 6+1 variables, then
\begin{align*}
\llp{u}_{S^0} := \sup_{(q, r) \ \text{admissible}} \llp{u}_{L^{q}_tL^r_xL^2_y} = \sup_{(q, r) \ \text{admissible}} \Lp{\llp{u(t, x, \cdot)}_{L^2_y}}{L^{q}_tL^r_x}.
\end{align*}

The Schr\"odinger group satisfies the dispersive estimates:
\begin{align}\label{disp-est}
\llp{e^{it\lapl} f}_{L^r_x} \lesssim |t|^{-(\frac{3}{2}-\frac{3}{r})} \llp{f}_{L^{r'}_x} \ \ \text{ for }\ \ 2\le r\le \infty.
\end{align}
From \eqref{disp-est}, we can deduce the standard Strichartz estimates: suppose $(q, r)$ and $(\tilde q, \tilde r)$ are 
admissible pairs then it follows
\begin{subequations}\label{strichartz}
\begin{align}
\llp{e^{it\lapl}f}_{L^q_tL^r_x}\lesssim&\ \llp{f}_{L^2_x}\label{h-strich}\\
\Lp{\int^t_0 ds\ e^{i(t-s)\lapl}g(s)}{L^q_tL^r_x} \lesssim&\ \llp{g}_{L^{\tilde q'}_tL^{\tilde r'}_x} \label{in-strich}.
\end{align}
The case $(q, r) = (2, 6)$ is called the endpoint Strichartz estimates; see \cite{KeelTao}. See \cite{Tao} for an excellent account of the rudimentary facts of dispersive PDEs.
\end{subequations}

\begin{prop}[a-priori estimates]\label{global}
Let $\phi$ be a solution to \eqref{htnls} and $\phi_0 \in \dot H^{\frac{1}{2}}$, then we have the estimate
\begin{align}
\llp{D^\frac{1}{2}\phi}_{S^0} \lesssim \llp{\phi_0}_{\dot H^{\frac{1}{2}}_x} + \llp{v}_{L^1_x}\llp{D^{\frac{1}{2}}\phi}_{S^0}^3
\end{align}
which is independent of $N$. If $\llp{\phi_0}_{\dot H^{\frac{1}{2}}_x}$ is sufficiently small then we obtain the estimate
\begin{align}
\llp{D^\frac{1}{2}\phi}_{S^0} \lesssim 1
\end{align}
which depends only on $\llp{\phi_0}_{\dot H^{\frac{1}{2}}_x}$ and independent of $N$. 
Moreover, by the Sobolev inequality, we have the estimate 
\begin{align}\label{scattering-norm}
\llp{\phi}_{L^5_{t}L^5_x} \lesssim 1. 
\end{align}
\end{prop}

\begin{proof}
We begin by differentiating \eqref{htnls} 
\begin{equation}
\begin{aligned}
&\frac{1}{i}\bd_t D^\frac{1}{2}\phi -\lapl D^\frac{1}{2}\phi + (v_N\ast |\phi|^2)\cdot D^\frac{1}{2}\phi
+ (v_N\ast D^\frac{1}{2}|\phi|^2)\cdot \phi\\
& + \text{ ``lower order" terms} =0.
\end{aligned}
\end{equation}
Applying the dual $L^2L^{\frac{6}{5}}$--endpoint Strichartz estimate \eqref{in-strich}
and the fractional Leibniz rule, we obtain the estimate
\begin{align*}
\llp{D^\frac{1}{2}\phi}_{S^0} \lesssim&\ \llp{e^{it\lapl} D^{\frac{1}{2}}\phi}_{S^0}+\llp{(v_N\ast |\phi|^2)\cdot D^\frac{1}{2}\phi}_{L^2_t L^{\frac{6}{5}}_x}\\
& +\llp{(v_N\ast D^\frac{1}{2} |\phi|^2)\cdot \phi}_{L^2_t L^{\frac{6}{5}}_x}\\
\lesssim&\ \llp{\phi_0}_{\dot H^{\frac{1}{2}}_x} 
+ \llp{v_N\ast |\phi|^2}_{L^2_tL^3_x}\llp{D^\frac{1}{2}\phi}_{L^\infty_tL^2_x}\\
& +\llp{v_N\ast D^\frac{1}{2}|\phi|^2}_{L^2_tL^2_x}\llp{\phi}_{L^\infty_tL^3_x}.
\end{align*}
For the first forcing term, we apply Sobolev and Young's inequalities to get 
\begin{align*}
 \llp{v_N\ast |\phi|^2}_{L^2_tL^3_x}\llp{D^\frac{1}{2}\phi}_{L^\infty_tL^2_x} 
 \lesssim&\  \llp{v}_{L^1_x}\llp{\phi}^2_{L^4_tL^6_x}\llp{D^\frac{1}{2}\phi}_{L^\infty_tL^2_x}\\
 \lesssim&\ \llp{v}_{L^1_x} \llp{D^\frac{1}{2}\phi}_{L^4_tL^3_x}^2\llp{D^\frac{1}{2}\phi}_{L^\infty_tL^2_x}\\
 \lesssim&\ \llp{v}_{L^1_x}\llp{D^\frac{1}{2}\phi}_{S^0}^3.
\end{align*}
The other term can be handled in a similar fashion. 
\end{proof}

As an immediate corollary of Proposition \ref{global}, we have
\begin{cor}[Uniform in $N$ global well-posedness]\label{gwp}
Let $v \in C^\infty_c(\rr^3)$. Then there exists $\varepsilon=\varepsilon(\llp{v}_{L^1_x})>0$, independent of $N$, such that for 
any $\varphi_0 \in \{ \varphi \in \dot H^{\frac{1}{2}}_x \mid \llp{\varphi}_{\dot H^{\frac{1}{2}}_x}<\varepsilon\}$
there exists a unique solution to \eqref{htnls} with initial data $\varphi_0$ satisfying 
$\varphi_t \in C([0, \infty)\rightarrow \dot H^{\frac{1}{2}}_x)\cap S^0$. 
\end{cor}

\begin{remark}\label{remark-critical-gwp}
The proof of Corollary \ref{gwp} is standard in the literature for showing small data global well-posedness of critical equations. See Remark 4.5 in \cite{staff} for a complete proof of the statement. 
\end{remark}

\begin{remark}\label{morawetz}
In Proposition \ref{global}, the uniform in $N$ control of the $L^5_tL^5_x$--norm of $\phi$ plays an important role in the following analysis for propagating the regularity of solutions to the family of Hartree equations.  This should
be compare with the  $L^4_tL^4_x$ a-priori estimate (also called the interaction Morawetz estimate) obtained in Proposition 3.1 of \cite{GM2}; the proof of the interaction Morawetz estimate in \cite{GM2} relies heavily on the positivity of
the interaction potential, but it does not require any smallness condition on the initial data which greatly differs from the above result.
\end{remark}

\begin{prop}[Propagation of Sobolev Regularity]\label{sobolev}
Let $\phi$ be a solution to \eqref{htnls} as in Corollary \ref{gwp} with  initial data $\phi_0\in \dot H^s$. Then there exists $C_s$ depending only on $\llp{\phi_0}_{H^s}$ 
such that the estimate
\begin{subequations}
\begin{align}\label{propagate}
\llp{\phi(t, \cdot)}_{\dot H^s} \le C_s
\end{align}
holds uniformly in $N$ and time. As an immediate consequence, we have that
\begin{align}\label{propagate2}
\llp{\bd_t^j\phi(t, \cdot)}_{\dot H^s_x} \le C_{s, j}
\end{align}
\end{subequations}
which follows immediately from  differentiating \eqref{htnls} and repeating the argument of \eqref{propagate}.
\end{prop}

\begin{proof}
By \eqref{scattering-norm}, we can split $[0, \infty)$ into finitely many intervals $I_k$ where 
\begin{align*}
\llp{\phi}_{L^5_t(I_k)L^5_x} \le \varepsilon.
\end{align*}
The choice of $\varepsilon$ will be determine later. Again,  differentiating \eqref{htnls}  yields
\begin{align*}
&\frac{1}{i}\bd_t D^s\phi -\lapl D^s\phi + (v_N\ast |\phi|^2)\cdot D^s\phi\\
& + \text{ ``similar or lower order" terms} =0.
\end{align*}
On the first interval $I_1$, we use the $L^{\frac{10}{3}}_tL^{\frac{10}{3}}_x$ Strichartz estimates and Young's inequality to get
\begin{align*}
\llp{D^s\phi}_{L^{\frac{10}{3}}_t(I_1)L^{\frac{10}{3}}_x} \le&\  C\llp{\phi_0}_{\dot H^s}+C\llp{(v_N\ast|\phi|^2) D^s\phi}_{L^{\frac{10}{7}}_t(I_1)L^{\frac{10}{7}}_x}\\
\le&\ C_1\llp{\phi_0}_{\dot H^s}+C_2\llp{v_N\ast|\phi|^2}_{L^\frac{5}{2}_t(I_1)L^\frac{5}{2}_x}\llp{D^s\phi}_{L^{\frac{10}{3}}_t(I_1)L^{\frac{10}{3}}_x}\\
\le&\ C_1\llp{\phi_0}_{\dot H^s}+C_2\llp{v}_{L^1_x}\llp{\phi}_{L^5_t(I_1)L^5_x}^2\llp{D^s\phi}_{L^{\frac{10}{3}}_t(I_1)L^{\frac{10}{3}}_x}.
\end{align*}
If we choose $\varepsilon$ such that $C_2\llp{v}_{L^1}\varepsilon^2<\frac{1}{2}$, then it follows
\begin{align}
\llp{D^s\phi}_{L^{\frac{10}{3}}_t(I_1)L^{\frac{10}{3}}_x} \le 2C_1\llp{\phi_0}_{\dot H^s}.
\end{align}
Repeating the proof for the remaining finite number of intervals gives us
\begin{align}
\llp{D^s\phi}_{L^{\frac{10}{3}}_tL^{\frac{10}{3}}_x} \le C\llp{\phi_0}_{\dot H^s}.
\end{align}
Finally, applying Strichartz one last time yields
\begin{align*}
\llp{\phi(t, \cdot)}_{\dot H^s} \leq&\ C\llp{\phi_0}_{\dot H^s}+C\llp{(v_N\ast|\phi|^2) D^s\phi}_{L^{\frac{10}{7}}_t(I_1)L^{\frac{10}{7}}_x}\\
\le&\ C_1\llp{\phi_0}_{\dot H^s}+C_2\llp{v}_{L^1_x}\llp{\phi}_{L^5_tL^5_x}^2\llp{D^s\phi}_{L^{\frac{10}{3}}_t(I_1)L^{\frac{10}{3}}_x}\le C_s
\end{align*}
which holds uniformly in $t$ and $N$. 
\end{proof}

\subsection{Decay Estimates}
In this subsection,  we prove the uniform in $N$  decay estimates for $\phi_t$ following 
the approach given in \cite{GM2}, which is in the spirit of \cite{SL}. Let us first make a note on the notation 
used in this section. The notation $\alpha\pm$ means $\alpha\pm\varepsilon$ for
 some fixed $0<\varepsilon\ll 1$. The slight analytic gymnastic introduced by using this notation is a consequence of the fact that  we do not have the endpoint Sobolev estimate.  
\begin{prop}\label{decay} Suppose $\phi_0 \in W^{k, 1}_x$ for some sufficiently large $k$. 
Let $\phi$ be a solution to \eqref{htnls} with small $\dot H^{\frac{1}{2}}_x$ data $\phi_0$. Then we have the decay estimate 
\begin{align}
\llp{\phi(t, \cdot)}_{L^\infty_x} \lesssim \frac{1}{1+t^{\frac{3}{2}}}
\end{align}
which only depends on $\llp{\phi_0}_{W^{k, 1}_x}$ and holds uniformly in $N$. 
\end{prop}
\begin{remark}\label{decay-NLS}
Proposition \ref{decay} holds for global smooth solutions to the focusing NLS. More precisely,  if $\phi_0$ is smooth and whose norm is below some ground state threshold, then the solution is
 global and its regularity is propagated by Proposition \ref{sobolev}. Hence the proof of Proposition \ref{decay} shows that $\phi$ also satisfies the decay estimate. This point is relevant
 for Theorem \ref{DerFNLS}.
\end{remark}

Let us first prove the following lemmas.
\begin{lemma}\label{decay-0}
Assuming the same conditions as in Proposition \ref{decay}. Then
 $\llp{\phi(t, \cdot)}_{L^\infty_x}\rightarrow 0$ as $t \rightarrow \infty$. 
\end{lemma}

\begin{proof}
By \eqref{scattering-norm} and Proposition \ref{sobolev}, we have the estimates
\begin{align}
\llp{\phi}_{L^5_tL^5_x} \leq C \ \ \text{ and } \ \ \llp{\phi}_{C^k(\rr^3\times \rr)} \le C_{k} \ \text{ for all } k \in \nn
\end{align}
where the latter follows from Sobolev embedding applied to \eqref{propagate2}. By the non-sharp version of the Sobolev embedding, we see that
\begin{align*}
\llp{\phi^2}_{L^p_t([n, n+1])L^p_x} \le&\ \llp{\grad_{t, x}(\phi^2)}_{L^5_t([n, n+1])L^5_x} \\
\le&\ 2\llp{\phi}_{L^5_t([n, n+1])L^5_x}\llp{\grad_{t, x} \phi}_{L^\infty_t([n, n+1])L^\infty_x}
\end{align*}
for $5<p<\infty$. In particular, it follows that as $n\rightarrow \infty$ we see that $\llp{\phi}_{L^{2p}_t([n, n+1])L^{2p}_x}\rightarrow 0$. Applying the argument again, we see that
\begin{align*}
\llp{\phi^2}_{L^\infty_t([n, n+1])L^\infty_x}\le&\  \llp{\grad_{t, x}(\phi^2)}_{L^{11}_t([n, n+1])L^{11}_x}\\
\le&\ 2\llp{\phi}_{L^{11}_t([n, n+1])L^{11}_x}\llp{\grad_{t, x} \phi}_{L^\infty_t([n, n+1])L^\infty_x}\rightarrow 0
\end{align*}
as $n\rightarrow \infty$, which yields the desired result. 
\end{proof}

\begin{lemma}\label{decay-1}
Assuming the same conditions as in Proposition \ref{decay}. Then there exists $k \in L^1([0, \infty))$  and $\delta>0$ such that
\begin{align}
\llp{e^{i(t-s)\lapl}((v_N\ast |\phi|^2)\cdot \phi(s))}_{L^\infty_x} \leq k(t-s) \llp{\phi(s, \cdot)}^{1+\delta}_{L^\infty_x}.
\end{align}
\end{lemma}

\begin{proof}
Using the $L^\infty L^1$ -decay estimate \eqref{disp-est} and conservation of mass of \eqref{htnls}, we have that
\begin{equation}\label{est-dec-1}
\begin{aligned}
\llp{e^{i(t-s)\lapl}((v_N\ast |\phi|^2)\cdot \phi(s))}_{L^\infty_x} \lesssim&\ \frac{1}{|t-s|^{\frac{3}{2}}}\llp{(v_N\ast |\phi|^2)\cdot \phi(s)}_{L^1_x}\\
\lesssim&\   \frac{1}{|t-s|^{\frac{3}{2}}}\llp{v}_{L^1_x}\llp{\phi}_{L^2_x}^2\llp{\phi(s, \cdot)}_{L^\infty_x}.
\end{aligned}
\end{equation}
On the other hand, applying Sobolev embedding, $L^{3+}L^{\frac{3}{2}-}$-- decay estimate \eqref{disp-est} and interpolation yields
\begin{align}\label{est-dec-2}
\llp{e^{i(t-s)\lapl}((v_N\ast|\phi|^2)\cdot \phi(s)}_{L^\infty_x} &\lesssim\ \llp{\grad e^{i(t-s)\lapl}((v_N\ast|\phi|^2)\cdot \phi(s))}_{L^{3+}_x} \nonumber\\
&\lesssim\ \frac{1}{|t-s|^{\frac{1}{2}+}}\llp{\grad\phi}_{L^2_x}
\llp{\phi}_{L^{12-}_x}^2\\
&\lesssim\  \frac{1}{|t-s|^{\frac{1}{2}+}}\llp{\grad\phi}_{L^2_x}\llp{\phi}_{L^2_x}^{\frac{1}{3}+}\llp{\phi}^{\frac{5}{3}-}_{L^\infty_x}.\nonumber\\
&\lesssim\  \frac{1}{|t-s|^{\frac{1}{2}+}}\llp{\phi}^{\frac{5}{3}-}_{L^\infty_x}\nonumber
\end{align}
where the last inequality follows from conservation of mass. 
In the case $|t-s|<1$, we could simply take $k(t-s) = |t-s|^{\frac{1}{2}+}$. In the case $|t-s|\geq 1$, we interpolate estimates \eqref{est-dec-1} and \eqref{est-dec-2}.
\end{proof}

\begin{proof}[Proof of Proposition \ref{decay}]
 Let $\phi_0$ be sufficiently smooth and write (for $t>0$)
\begin{align}\label{phi-sol}
 \phi(t)= e^{it\lapl}\phi_0-i \left[\int^{\frac{t}{2}}_0+\int^{t}_{\frac{t}{2}}\right] d\tau\ e^{i(t-\tau)\lapl}(v_N\ast |\phi(\tau)|^2)\phi(\tau).
\end{align}
Taking the $L^\infty$ norm of \eqref{phi-sol} yields
\begin{align}
 \llp{\phi(t)}_{L^\infty_x} \lesssim \frac{\llp{\phi_0}_{L^1_x} }{t^{\frac{3}{2}}}+\left[ \int^{\frac{t}{2}}_0+\int^t_{\frac{t}{2}}\right] \llp{e^{i(t-\tau)\lapl}(v_N\ast|\phi|^2)\phi(\tau)}_{L^\infty_x}\ d\tau
\end{align}
where the first term is a consequence of the $L^\infty L^1$-decay estimate.
For the first part of the second term, we apply the $L^\infty L^1$-decay estimate, Young's convolution inequality, and conservation of mass to get
\begin{align*}
 \int^{\frac{t}{2}}_0 d\tau\ \llp{e^{i(t-\tau)\lapl}(v_N\ast|\phi|^2)\phi(\tau)}_{L^\infty_x}
&\leq \int^{\frac{t}{2}}_0d\tau\ \frac{\llp{(v_N\ast|\phi|^2)\phi(\tau)}_{L^1_x}}{|t-\tau|^{\frac{3}{2}}}\\
&\lesssim \frac{1}{t^{\frac{3}{2}}} \int^{\frac{t}{2}}_0  d\tau\ \llp{\phi(\tau)}_{L^\infty_x}.
\end{align*}
Lastly, by Lemma \ref{decay-1}, there exists $k \in L^1([0, \infty])$ and $\delta>0$ such that
\begin{align*}
 \int^t_{\frac{t}{2}} d\tau\ \llp{e^{i(t-\tau)\lapl}(v_N\ast|\phi|^2)\phi(\tau)}_{L^\infty_x}
\lesssim \int^t_{\frac{t}{2}} d\tau\ k(t-\tau) \llp{\phi(\tau)}^{1+\delta}_{L^\infty_x}.
\end{align*}
Combining all the estimates yields
\begin{align*}
 \llp{\phi(t)}_{L^\infty_x}\lesssim \frac{\llp{\phi_0}_{L^1_x}}{t^{\frac{3}{2}}}+\int^{\frac{t}{2}}_0 d\tau \frac{\llp{\phi(\tau)}_{L^\infty_x} }{t^{\frac{3}{2}}}+\int^t_{\frac{t}{2}} d\tau\ k(t-\tau) \llp{\phi(\tau)}^{1+\delta}_{L^\infty_x}
\end{align*}
which holds for all $t>0$. 

Since we care about large time behavior we may assume $t\geq 1$. 
In particular, we get the equivalent estimate
\begin{align}\label{estimate1}
 \llp{\phi(t)}_{L^\infty_x}\lesssim \frac{\llp{\phi_0}_{L^1_x}}{1+t^{\frac{3}{2}}}+ \int^{\frac{t}{2}}_0  d\tau\ \frac{\llp{\phi(\tau)}_{L^\infty_x}}{1+t^{\frac{3}{2}}}+\int^t_{\frac{t}{2}}d\tau\ k(t-\tau) \llp{\phi(\tau)}^{1+\delta}_{L^\infty_x}.
\end{align}
Multiply \eqref{estimate1} by $1+t^{\frac{3}{2}}$ yields
\begin{equation}\label{q1}
\begin{aligned}
(1+t^{\frac{3}{2}})\llp{\phi(t)}_{L^\infty_x}
\lesssim&\ \llp{\phi_0}_{L^1_x}+ \int^{\frac{t}{2}}_0 d\tau\ \llp{\phi(\tau)}_{L^\infty_x} \\
&\ +\sup_{\frac{t}{2} \leq \tau \leq t} (1+\tau^{\frac{3}{2}})\llp{\phi(\tau)}^{1+\delta}_{L^\infty_x}
\end{aligned}
\end{equation}
since $k \in L^1([0, \infty))$. Next, by Lemma \ref{decay-0}, there exists $T>0$ such that 
\begin{equation}\label{q2}
\begin{aligned}
(1+t^{\frac{3}{2}})\llp{\phi(t)}_{L^\infty_x}\leq&\  c\llp{\phi_0}_{L^1_x}+ c\int^{\frac{t}{2}}_0 d\tau\ \llp{\phi(\tau)}_{L^\infty_x}\\
&\ +\frac{1}{2}\sup_{\frac{t}{2} \leq \tau \leq t} (1+\tau^{\frac{3}{2}})\llp{\phi(\tau)}_{L^\infty_x}
\end{aligned}
\end{equation}
whenever $t \geq 2T$ for some constant $c>0$. 

Let us define the quantities
\begin{subequations}
\begin{align}
M(t) :=& \sup_{T \leq s \leq t} (1+s^{\frac{3}{2}})\llp{\phi(s)}_{L^\infty_x}\\
C_T :=& \sup_{0\leq  s \leq 2T} (1+s^{\frac{3}{2}})\llp{\phi(s)}_{L^\infty_x}.
\end{align} 
\end{subequations}
By definition, we see that  $M(t) \le C_T$ when  $T\le t \le 2T$. In the case $2T\le t$, it follows from \eqref{q2} that
\begin{align*}
 (1+t^{\frac{3}{2}})\llp{\phi(t)}_{L^\infty_x} \leq C_1+ c\int^{\frac{t}{2}}_T \frac{M(\tau)}{1+\tau^{\frac{3}{2}}}\ d\tau +\frac{1}{2} M(t).
\end{align*}
where $C_1$ depends on $T$.
Note we also have the estimate 
\begin{align}
 (1+s^{\frac{3}{2}})\llp{\phi(s)}_{L^\infty_x} \leq \max\left(C_1 + c\int^{\frac{t}{2}}_T \frac{M(\tau)}{1+\tau^{\frac{3}{2}}}\ d\tau +\frac{1}{2}M(t), C_T\right).
\end{align}
for all $T<s <t$.
Hence it follows
\begin{align}
 M(t) \leq \max\left(C_1+ c\int^{\frac{t}{2}}_T \frac{M(\tau)}{1+\tau^{\frac{3}{2}}}\ d\tau +\frac{1}{2}M(t), C_T\right)
\end{align}
for all $t \geq T$. Then, by Gronwall's inequality, we get that
\begin{align*}
 M(t) \lesssim \max\left(\exp\left(\int^t_0 \frac{d\tau}{1+\tau^{\frac{3}{2}}} \right), C_T\right) \lesssim 1.
\end{align*}
Thus, we have established the desired result
\begin{align*}
 \sup_{0 \leq s \leq t} (1+s^{\frac{3}{2}})\llp{\phi(s)}_{L^\infty_x} \lesssim \max(M(t), C_T) \lesssim 1.
\end{align*}
\end{proof}

\begin{cor}\label{high-deriv-decay}
Assume the conditions of Proposition \ref{decay}. Then there exists constants $C_1$ depending only on 
$\llp{\phi_0}_{W^{k, 1}}$  and $\llp{\bd_t\phi_0}_{W^{k, 1}}$ and $C_2$ depending only on $\llp{\phi_0}_{W^{k, 1}}$ such that
\begin{subequations}
\begin{align}
  \llp{\bd_t\phi(t,\cdot)}_{L^\infty_x} \le&\ \frac{C_1}{1+t^{\frac{3}{2}}},\\
 \llp{\bd \phi(t,\cdot)}_{L^\infty_x} \le&\ \frac{C_2}{1+t^{\frac{3}{2}}}.
\end{align}
\end{subequations}
In fact, by iterating the proof of the above estimates, we could show that similar decay estimates hold for $\bd_t^j\bd^\alpha \phi$ for arbitrary $\alpha$ and $j$ provided the data is sufficiently smooth. 
\end{cor}

\begin{proof}
We begin by taking the one partial derivative (spatial or time) of (\ref{htnls})
\begin{align*}
 \frac{1}{i}\frac{\bd}{\bd t} \bd \phi - \lapl \bd \phi +\bd(v_N\ast |\phi|^2)\phi = 0.
\end{align*}
Applying the $L^\infty L^1$-decay estimate yields $(t\geq 1)$
\begin{align*}
 \llp{\bd_t\phi(t)}_{L^\infty_x} \lesssim&\ \frac{\llp{\bd \phi_0}_{L^1_x}}{t^{\frac{3}{2}}}+ \int^{\frac{t}{2}}_0 d\tau\ \llp{e^{i(t-\tau)\lapl}\bd(v_N\ast |\phi|^2)\phi(\tau)}_{L^\infty_x}\ \\
&\ +  \int^{t}_{\frac{t}{2}}d\tau\ \llp{e^{i(t-\tau)\lapl}\bd(v_N\ast |\phi|^2)\phi(\tau)}_{L^\infty_x}\ .
\end{align*}
For the first integral, we again apply the $L^\infty L^1$-decay estimate, Proposition \ref{sobolev}, and Proposition 
\ref{decay} to get
\begin{align*}
\int^{\frac{t}{2}}_0 d\tau\ \llp{e^{i(t-\tau)\lapl}&\bd(v_N\ast |\phi|^2)\phi(\tau)}_{L^\infty_x}
\lesssim \int^{\frac{t}{2}}_0d\tau\ \frac{\llp{\bd(v_N\ast |\phi|^2)\phi(\tau)}_{L^1_x}}{|t-\tau|^{\frac{3}{2}}}\\
\lesssim&\ \frac{1}{1+t^{\frac{3}{2}}}\int^{\frac{t}{2}}_0 d\tau\ \llp{\phi(\tau)}_{L^\infty_x}\llp{\phi(\tau)}_{L^2_x}\llp{\bd\phi(\tau)}_{L^2_x}\\
\lesssim&\ \frac{1}{1+t^{\frac{3}{2}}}\int^{\frac{t}{2}}_0 d\tau\ \frac{d \tau}{1+\tau^{\frac{3}{2}}} \lesssim \frac{1}{1+t^{\frac{3}{2}}}.
\end{align*}
\

For the second integral, we use Sobolev embedding and $L^{3+}L^{\frac{3}{2}-}$--decay estimate to obtain the bound
\begin{align*}
 &\int^t_{\frac{t}{2}} d\tau\ \llp{e^{i(t-\tau)\lapl}\bd[(v_N\ast |\phi|^2)\phi(\tau)]}_{L^\infty_x}\\
  &\lesssim\ \int^t_{\frac{t}{2}} d\tau\ \llp{\grad_x e^{i(t-\tau)\lapl}\bd[(v_N\ast |\phi|^2)\phi(\tau)]}_{L^{3+}_x}\\
 &\lesssim\ \int^{t}_{\frac{t}{2}}d\tau\ \frac{1}{|t-\tau|^{\frac{1}{2}+}} \llp{\grad\bd\phi}_{L^2_x}\llp{\phi}_{L^2_x}
\llp{\phi}^{\frac{5}{3}-}_{L^\infty_x}\\ 
&\ \ \ +\int^t_{\frac{t}{2}}d\tau\ \frac{1}{|t-\tau|^{\frac{1}{2}+}} \llp{\bd\phi}_{L^2_x}^{\frac{1}{3}+}\llp{\bd\phi}_{L^\infty_x}^{\frac{2}{3}-}\llp{\grad\phi}_{L^2_x}
\llp{\phi}_{L^\infty_x}.
\end{align*}
Note the last inequality is a consequence of H\"older inequalities and space interpolation. Then, by Proposition \ref{sobolev} and Propostion \ref{decay}, it follows that
\begin{align*}
  &\int^t_{\frac{t}{2}} \llp{e^{i(t-\tau)\lapl}\bd(v_N\ast |\phi|^2)\phi(\tau)}_{L^\infty_x}\ d\tau
\lesssim \int^{t}_{\frac{t}{2}}\frac{1}{|t-\tau|^{\frac{1}{2}+}}\llp{\phi(\tau)}_{L^\infty_x}\ d\tau\\
&\lesssim\ \frac{1}{1+t^{\frac{3}{2}}} \int^{t}_{\frac{t}{2}} \frac{1}{|t-\tau|^{\frac{1}{2}+}}\ d\tau \lesssim \frac{1}{1+t^{\frac{3}{2}}}.
\end{align*}
\end{proof}
Applying linear interpolation, we get the following corollary. 
\begin{cor}\label{L3-L4}
Assume the conditions of Proposition \ref{decay}. Then there exists a constant $C$ depending only on 
$\llp{\phi_0}_{W^{k, 1}}$  and $\llp{\bd_t\phi_0}_{W^{k, 1}}$ for $k$ sufficiently large such that
\begin{subequations}
\begin{align}
\llp{\bd^\alpha\phi(t, \cdot)}_{L^4_x}+\llp{\bd_t\phi(t, \cdot)}_{L^4_x} \leq&\ \frac{C}{1+t^\frac{3}{4}},\\
\llp{\bd^\alpha\phi(t, \cdot)}_{L^3_x}+\llp{\bd_t\phi(t, \cdot)}_{L^3_x} \leq&\ \frac{C}{1+t^\frac{1}{2}}.
\end{align}
\end{subequations}
\end{cor}

\section{Estimates for the Pair Excitations}

The major results of this section are Proposition \ref{pexc-1} and Proposition \ref{s2-collapsing}. Let us define the shorthand notation $\ch (k)= c_1= \delta +p_1, \sh (k)=s_1=u$, and also $\ch(2k):= c_2= \delta+p_2, \sh (2k):=s_2$. Let us also recall the equations for
$\sh(2k)$ and $\ch(2k)$
\begin{subequations}\label{s2-system}
\begin{align}
&\vect{S}(s_2) =\ 2m+m\circ p_2 +\overline{p_2}\circ m, \label{s2-equation}\\
&\vect{W}(\overline{p_2})=\ m\circ \overline{s_2}- s_2\circ \overline{m}\\
&\text{with } s_2(0, \cdot) =s_{2, 0}\ p_2(0, \cdot) = p_{2, 0} \nonumber
\end{align}
\end{subequations}
where $m(t, x, y) = -v_N(x-y)\phi(t, x)\phi(t, y)$. 

\begin{prop}\label{pexc-1}
Assume $\phi_0 \in W^{k, 1}$ for $k$ sufficiently large. Then we have the following estimates
\begin{subequations}
\begin{align}
\llp{\grad_{x+y}^js_2(t, \cdot)}_{L^2_{x, y}}\lesssim&\ \llp{\grad_{x+y}^js_{2, 0}}_{L^2_{x, y}}+ \llp{\grad_{x+y}^jp_{2, 0}}_{L^2_{x, y}}\ \text{ for } j\ge 0, \label{s2-est}\\
\llp{\grad_{x+y}^jp_2(t, \cdot)}_{L^2_{x, y}} \lesssim&\  \llp{\grad_{x+y}^js_{2, 0}}_{L^2_{x, y}}+ \llp{\grad_{x+y}^jp_{2, 0}}_{L^2_{x, y}}\  \text{ for } j\ge 0, \label{p2-est}\\
\sup_x\llp{s_2(t, x, \cdot)}_{L^2_y} \lesssim&\  \sum^2_{j=0}(\llp{\grad_{x+y}^js_{2, 0}}_{L^2_{x, y}} +\llp{\grad_{x+y}^jp_{2, 0}}_{L^2_{x, y}})\label{L2Linfty}
\end{align}
\end{subequations}
where each estimate only depends on $\llp{\phi_0}_{W^{k, 1}}$ and independent of $N$.  Here, we have adopted the notation $\grad_{x\pm y} := \grad_x\pm \grad_y$.
\end{prop}

As an immediate corollary of Proposition \ref{pexc-1}, we have the useful result.
\begin{cor}\label{main-cor}
We have the following uniform in time estimates
\begin{subequations}
\begin{align}
\llp{s_1(t, \cdot, \cdot)}_{L^2_{x, y}} +\sup_x\llp{s_1(t, x, \cdot)}_{L^2_{y}} \lesssim&\ 1\\
\llp{p_1(t, \cdot, \cdot)}_{L^2_{x, y}} 
+
\sup_x\llp{p_1(t, x, \cdot)}_{L^2_{y}} \lesssim&\ 1 \label{p1-est}
\end{align}
where the estimates only depend on the initial data. 
Moreover, by interpolation, we also have the estimate 
\begin{align}
\Lp{\llp{s_1(t, \cdot, \cdot)}_{L^2_{y}} }{L^4_x}+\Lp{\llp{p_1(t, \cdot, \cdot)}_{L^2_{y}} }{L^4_x} \lesssim 1.
\end{align}
\end{subequations}
\end{cor}

\begin{proof}
See proof of Corollary 5.3 in \cite{Chong}.
\end{proof}

\begin{prop}\label{s2-collapsing}
Let $s_2, p_2$ be smooth solutions to \eqref{s2-system}. Then, for every  fixed $\varepsilon>0$, we have the estimates
\begin{subequations}
\begin{align}
&\llp{\langle \grad\rangle s_2}_{L^\infty_tL^2_{x, y}} 
\lesssim\  \llp{s_{2, 0}}_{\dot H^1(\rr^6)}+\llp{p_{2, 0}}_{\dot H^1(\rr^6)}+ N^{\frac{\beta}{2}(1+2\varepsilon)}\label{L2-deriv-s2}\\
&\llp{\langle \grad\rangle p_2}_{L^\infty_tL^2_{x, y}} 
\lesssim\   \llp{s_{2, 0}}_{\dot H^1(\rr^6)}+\llp{p_{2, 0}}_{\dot H^1(\rr^6)}+  N^{\frac{\beta}{2}(1+2\varepsilon)} \label{L2-deriv-p2}\\
&\sup_z\llp{s_2(t, x, x+z)}_{L^2_t([0, T])L^2_x} \lesssim\ C_0(T, N)\label{s2-collapsing-est}  
\end{align}
\end{subequations}
where 
\begin{align}
C_0(T, N) :=  N^{\beta (1+\varepsilon)} +\left(\llp{s_{2, 0}}_{\dot H^\frac{1}{2}(\rr^6)}+N^{\frac{\beta}{4}(1+2\varepsilon)}\right)\sqrt{T}
\end{align}
and $\langle \cdot \rangle= \sqrt{1+|\cdot|^2}$ is the standard bracket notation. The $\grad$ in \eqref{L2-deriv-s2} and \eqref{L2-deriv-p2} refers to either
$x$ or $y$ derivative. 
\end{prop}

\begin{cor}\label{s1-collapsing}
We have the estimates
\begin{subequations}
\begin{align}
&\llp{\langle \grad\rangle s_1(t, \cdot, \cdot)}_{L^\infty_tL^2_{x, y}} 
\lesssim\ \llp{s_{2, 0}}_{\dot H^1(\rr^6)}+\llp{p_{2, 0}}_{\dot H^1(\rr^6)}+ N^{\frac{\beta}{2}(1+2\varepsilon)}\label{L2-deriv-s1}\\
&\llp{\langle \grad\rangle p_1(t, \cdot, \cdot)}_{L^\infty_tL^2_{x, y}} 
\lesssim\ \llp{s_{2, 0}}_{\dot H^1(\rr^6)}+\llp{p_{2, 0}}_{\dot H^1(\rr^6)}+ N^{\frac{\beta}{2}(1+2\varepsilon)} \label{L2-deriv-p1}\\
&\sup_z\llp{s_1(t, x, x+z)}_{L^2_t([0, T])L^2_x} \lesssim\  C_0(T, N). \label{s1-collapsing-est} 
\end{align}
\end{subequations}
\end{cor}

\begin{proof}
Recall that $s_1= \frac{1}{2}s_2\circ c^{-1}$, then it follows
\begin{align*}
\llp{\grad_xs_1}_{L^2_{x, y}} \le \frac{1}{2}\llp{c^{-1}_1}_\text{op}\llp{\grad_xs_2}_{L^2_{x,y}} 
\le C\llp{\grad_xs_2}_{L^2_{x,y}}
\end{align*}
since $c^{-1}$ is a bounded operator. 

Using the identity $\bar s_1\circ s_1 = p_1\circ p_1 + 2p_1$ then it follows
\begin{align}
\overline{ \grad_xs_1}\circ \grad_ys_1 = \grad_x p_1\circ \grad_yp_1 + 2\grad_x\grad_y p_1.
\end{align}
Note that $\grad_x\grad_y p_1$ is a positive trace class operator, then it follows
\begin{align}
\llp{\grad_x p_1}_{L^2_{x, y}} \le \llp{\grad_x s_1}_{L^2_{x, y}}.
\end{align}

Since $s_2 = 2s_1+s_1\circ p_1$, then it follows
\begin{align*}
&\llp{s_1(x, x+z)}_{L^2_t([0, T])L^2_x}\\
 &\le \llp{s_2(x, x+z)}_{L^2_t([0, T])L^2_x}+\llp{(s_1\circ p_1)(x, x+z)}_{L^2_t([0, T])L^2_x}.
\end{align*}
By \eqref{s2-collapsing-est} and Corollary \ref{main-cor}, we have 
\begin{align*}
&\llp{(s_1\circ p_1)(x, x+z)}_{L^2_t([0, T])L^2_x}^2\\
&\lesssim \int^T_0 dt\ \llp{s_1(t)}_{L^2_{x, y}}^2\llp{p_1(t)}_{L^2_{x, y}}^2 \lesssim T.
\end{align*}
Then the result follows. 
\end{proof}

\begin{remark}
Let us make a comment on the initial condition $k_0$. Following \cite{BdOS}, we expect the correlation structure of the ground state of the many-body system
is well approximated by the pair excitation function $k_0(x, y) = -N w(N(x-y))\phi_0(x)\phi_0(y)$
where $1-w(N x)$  is the solution to the  zero-energy scattering equation for the rescaled potential $N^2v(Nx)$. The function $w(x)$ is smooth near the origin and behaviors like $a_0 |x|^{-1}$ where $a_0$ is the scattering length
of $v$. In particular, by the Hardy-Littlewood-Sobolev inequality, we have that
\begin{align}
\int dxdy\ |k_0(x, y)|^2 \leq C_1 \int dxdy\ \frac{|\phi_0(x)|^2|\phi_0(y)|^2}{|x-y|^2} \le  C_2\llp{|\grad|^\frac{1}{2}\phi_0}_{L^2_x}^4
\end{align}
which means $\llp{\sh(k_0)}_{L^2(dxdy)} \leq C$; the same is true for $p_1(0)$. This should be compare with the conditions $(7)$ given in Theorem 1 of \cite{NaNa3}.
\end{remark}

\subsection{Proof of Proposition \ref{pexc-1}} To prove Proposition \ref{pexc-1}, we begin by proving a few preliminary lemmas.
\begin{lemma}\label{m-est}
Let $M(t, x, y) := -v_N(x-y)f(t, x) g(t, y)$. Then  we have the following estimates
\begin{subequations}
\begin{align}\label{1st-est}
 &\int \frac{|\widehat{M}(t, \xi, \eta)|^2}{(|\xi|^2+|\eta|^2)^2}\ d\xi d\eta \lesssim \llp{f(t, \cdot)}_{L^3_x}^2\llp{g(t, \cdot)}_{L^3_x}^2\\
 &\int_{|\xi-\eta|>1} \frac{|\bd_t\widehat{M}(t, \xi, \eta)|^2}{(|\xi|^2+|\eta|^2)^2}\ d\xi d\eta \label{2nd-est}\\
  &\lesssim \llp{\bd_tf(t, \cdot)}_{L^4_x}^2 \llp{g(t, \cdot)}_{L^4_x}^2+\llp{f(t, \cdot)}_{L^4_x}^2 \llp{\bd_tg(t, \cdot)}_{L^4_x}^2. \nonumber
\end{align}

\end{subequations}
\end{lemma}

\begin{proof}
The proof of \eqref{1st-est} can be found in \cite{GM2}. We shall focus on the proof of the second estimate. First, observe
\begin{align}
v_N(x-y)f(x)g(y)  = \int dz\ \delta(x-y-z) v_N(z) f(x)g(y)
\end{align}
then the Fourier transform of $\delta(x-y-z)f(x)g(y)$ is given by
\begin{equation}
\begin{aligned}
&\int dxdy\ e^{-i(x\cdot \xi+y\cdot \eta)}\delta(x-y-z)f(x)g(y)\\
&= e^{iz\cdot \eta}\int dx\ e^{-ix\cdot(\xi+\eta)}f(x)g(x-z) = e^{iz\cdot\eta}\widehat{f g_z}(t, \xi+\eta).
\end{aligned}
\end{equation}
In particular, it follows that
\begin{align*}
|\bd_t\widehat{M}(t, \eta, \xi)|^2 =&\ \left| \int dz\ e^{iz\cdot\eta}v_N(z)\widehat{\bd_t(fg_z)}(t, \xi+\eta) \right|^2\\
\lesssim&\ \llp{v}_{L^1_x} \int dz\  |v_N(z)| |\widehat{\bd_t(fg_z)}(t, \xi+\eta)|^2.
\end{align*}
Finally, we have that
\begin{align*}
\int_{|\xi-\eta|>1} \frac{|\bd_t\widehat{M}(t, \eta, \xi)|^2}{(|\eta|^2+|\xi|^2)^2}\ d\eta d\xi \lesssim&\  \int |v_N(z)|
\int_{|\xi-\eta|>1}\frac{|\widehat{\bd_t(fg_z)}(t, \eta+\xi)|^2}{(|\eta|^2+|\xi|^2)^2}\ d\eta d\xi dz\\
\lesssim&\  \int |v_N(z)|\int_{|\eta'|>1}\frac{|\widehat{\bd_t(fg_z)}(t, \xi')|^2}{(|\eta'|^2+|\xi'|^2)^2}\ d\eta' d\xi' dz\\
\lesssim&\ \int |v_N(z)| |\widehat{\bd_t(fg_z)}(t, \xi')|^2\ d\xi' dz\\
 \lesssim&\ \llp{\bd_t f(t, \cdot)}_{L^4_x}^2\llp{g(t, \cdot)}_{L^4_x}^2+\text{ sym. term}.
\end{align*}
\end{proof}

\begin{lemma}\label{s0-lem}
Let $s^0_a$ be the solution to 
\begin{align}\label{s0-eq}
\left( \frac{1}{i}\frac{\bd}{\bd t} - \lapl_x-\lapl_y\right)s_a^0(t, x, y) = 2m(t, x, y),\ \ \ s^0_a(0, \cdot) = 0.
\end{align}
Then we have the uniform in $t$ estimate
\begin{align}\label{L2-s0-est}
\Lp{s^0_a(t, \cdot)}{L^2_{x, y}} \lesssim 1
\end{align}
where the estimate only depends on $\llp{\phi_0}_{W^{k, 1}}$. Similar estimates holds for $\grad_{x+y}^j s_a^0$. 
\end{lemma}
\begin{proof}
By Duhamel's principle, we have that
\begin{align*}
&\Lp{s^0_a(t, \cdot)}{L^2_{x, y}} =\ 2\Lp{\int^t_0e^{i(t-s)\lapl}m(s, \cdot)\ ds}{L^2_{x, y}}\\
&\lesssim\ \Lp{P_{|\xi-\eta|\leq 1}\int^t_0e^{i(t-s)\lapl}m(s, \cdot)\ ds}{L^2_{x, y}}+\Lp{P_{|\xi-\eta|> 1}\int^t_0e^{i(t-s)\lapl}m(s, \cdot)\ ds}{L^2_{x, y}}.
\end{align*}
Here, $P_A$ is a Littlewood-Paley projection (convolution) operator; more precisely, $\mathcal{F}(P_A f)(\xi) = \chi_A(\xi)\hat f(\xi).$ For the first term, we apply Minkowski's inequality to get
\begin{align*}
 &\bigg\|P_{|\xi-\eta|\leq 1}\int^t_0ds\ e^{i(t-s)\lapl}m(s, \cdot)\bigg\|_{L^2_{x, y}}\lesssim  \int^t_0 ds\ \left[\int_{|\xi-\eta|\leq 1}d\xi d\eta\ |\widehat m(s, \xi, \eta)|^2\right]^{\frac{1}{2}}\\
&\lesssim\ \int^t_0 ds\ \left[\int dz\ |v_N(z)|\int_{|\eta'|\leq 1} d\xi' d\eta' \ |\widehat{\phi\phi_z}(s,\xi')|^2\right]^{\frac{1}{2}}\lesssim \int^t_0 ds\ \llp{\phi(s, \cdot)}_{L^4_x}^2.
\end{align*}
By Corollary \ref{L3-L4}, we see that the first term is uniformly bounded in time. For the second term, we have
\begin{align*}
&\bigg\|P_{|\xi-\eta|> 1}\int^t_0ds\ e^{i(t-s)\lapl}m(s, \cdot)\bigg\|_{L^2_{x, y}} \\
&=\  \Lp{\chi_{|\xi-\eta|> 1}\int^t_0ds\ \bd_s e^{i(t-s)(|\eta|^2+|\xi|^2)}\frac{\widehat m(s,\xi, \eta)}{|\eta|^2+|\xi|^2}}{L^2_{x, y}}\\
&\lesssim\ \Lp{\frac{\widehat m(0,\xi, \eta)}{|\eta|^2+|\xi|^2}}{L^2_{\xi,\eta}} +\Lp{\frac{\widehat m(t, \xi, \eta)}{|\eta|^2+|\xi|^2}}{L^2_{\xi, \eta}}\\
&\quad +\Lp{\chi_{|\xi-\eta|>1}\int^t_0ds\ e^{i(t-s)(|\eta|^2+|\xi|^2)}
\frac{\bd_s\widehat m(s, \xi, \eta)}{|\eta|^2+|\xi|^2}}{L^2_{\xi, \eta}}.
\end{align*}
By Lemma \ref{m-est}, we see that the first two terms are bounded. For the last term, using Minkowski's and Lemma  \ref{m-est}, we have that
\begin{align*}
&\bigg\|\chi_{|\xi-\eta|>1}\int^t_0ds\ e^{i(t-s)(|\eta|^2+|\xi|^2)}
\frac{\bd_s\widehat m(s, \eta, \xi)}{|\eta|^2+|\xi|^2}\bigg\|_{L^2_{\xi,\eta}} \\
&\lesssim \int^t_0 ds\ \llp{\bd_t\phi(s, \cdot)}_{L^4_x}\llp{\phi(s, \cdot)}_{L^4_x}.
\end{align*}
By Corollary \ref{L3-L4}, the second term is also bounded uniformly in time. 

Since $[\grad_{x+y}, v_N(x-y)] = 0$, then we have the equation
\begin{align}
\left( \frac{1}{i}\frac{\bd}{\bd t} - \lapl_x-\lapl_y\right)\grad_{x+y}^js_a^0 =&\ -2v_N(x-y)\grad_{x+y}^j(\phi(x)\phi(y)).
\end{align}
Now, repeat the above argument yields the desired result.
\end{proof}

The following lemma will help us handle the ``potential" $V$.
\begin{lemma}\label{bounded-pot-lem}
Recall the definition of the potential $V(u)= ((v_N\ast |\phi|^2)(x)+(v_N\ast |\phi|^2)(y))u- (v_N\bar\phi\otimes\phi)\circ u-u\circ (v_N\bar\phi\otimes\phi)$. Let us also denote 
\begin{subequations}
\begin{equation}
\begin{aligned}
(\bd^k_x V)(u):=&\  ((v_N\ast \bd^k|\phi|^2)(x)+(v_N\ast |\phi|^2)(y))u\\
&\ - v_N\Pi^k(\bar\phi\otimes \phi)\circ u-u\circ v_N\Pi^k(\bar\phi\otimes \phi)
\end{aligned}
\end{equation}
and likewise for $(\bd^k_y V)(u)$ where $k\ge 0$ and 
\begin{align}
\Pi^k (\bar\phi\otimes \phi):= \sum^k_{j=0} \binom{k}{j} \overline{\bd^j\phi}\otimes \bd^{k-j}\phi.
\end{align} 
\end{subequations}
Then $\bd^kV(\cdot): L^2(\rr^6) \rightarrow L^2(\rr^6)$ is
a bounded operator and there exists a $C_k$, depending only on $k$ and independent of $N$, such that
\begin{align}
\llp{\bd^kV(u)}_{L^2_{x, y}} \le \frac{C}{1+t^3}\llp{u}_{L^2_{x, y}}.
\end{align}
\end{lemma}

\begin{proof}
By Young's inequality and Corollary \ref{high-deriv-decay}, we see that
\begin{align*}
\llp{(v_N\ast \bd^k|\phi|^2)(x) u}_{L^2_{x, y}} \le&\ \llp{v}_{L^1_x}\llp{ \bd^k|\phi|^2}_{L^\infty_x} \llp{u}_{L^2_{x, y}}
\le\ \frac{C}{1+t^3} \llp{u}_{L^2_{x, y}}.
\end{align*}
Similarly, for the other term, we have that 
\begin{align*}
&\llp{ v_N\Pi^k (\bar\phi\otimes \phi)\circ u}_{L^2_{x, y}}\\
&\le\ \sum^k_{j=0}\binom{k}{j}\int dz\ |v_N(z)| \Lp{\overline{\bd^j\phi(x)}\bd^{k-j}\phi(x-z)u(x-z, y)}{L^2_{x, y}}\\ 
 &\le\ \llp{v}_{L^1_x}\sum^k_{j=0}\binom{k}{j}\llp{\bd^j\phi}_{L^\infty_x}\llp{\bd^{k-j}\phi}_{L^\infty_x} \llp{u}_{L^2_{x, y}}\le\ \frac{C}{1+t^3} \llp{u}_{L^2_{x, y}}.
\end{align*}
\end{proof}

\begin{lemma}\label{sa-lem}
 Let $s_a$ be a solution to 
\begin{align}\label{sa-eq}
 \vect{S}(s_a) = 2m(t, x, y), \ \ s_a(0, \cdot) = 0.
\end{align}
Then 
\begin{align}\label{L2-sa-est}
 \Lp{s_a(t, \cdot)}{L^2_{x, y}} \lesssim 1
\end{align}
where the estimate depends only on the $\llp{\phi_0}_{W^{k, 1}}$. Similar estimates holds for $\grad_{x+y}^j s_a$.
\end{lemma}
\begin{proof}[Sketch of Proof]
We follow closely the proof of Lemma 4.5 in \cite{GM2}. The idea is to decompose the solution into two parts 
$ s_a = s^0_a + s^1_a$
where $s^0_a$ satisfies \eqref{s0-eq} and $s^1_a$ solves
\begin{align}\label{sa-eq2}
 \vect{S}(s_a^1) = -V(s_a^0).
\end{align}
By energy estimate, Lemma \ref{bounded-pot-lem}, and Lemma \ref{s0-lem}, we see that
\begin{align*}
\frac{d}{dt} \llp{s_a^1}^2_{L^2_{x, y}} \le&\ 2\llp{V(s_a^0)}_{L^2_{x, y}} \llp{s_a^1}_{L^2_{x, y}}\\
\le&\  \frac{C}{1+t^3} \llp{s_a^0}_{L^2_{x, y}} \llp{s_a^1}_{L^2_{x, y}} \le  \frac{C}{1+t^3}\llp{s_{2, 0}}_{L^2_{x, y}} \llp{s_a^1}_{L^2_{x, y}}.
\end{align*}
Finally, integrating in time yields the desired result. 

Differentiating \eqref{sa-eq2} with respect to $\grad_{x+y}$ yields
\begin{equation}
\begin{aligned}
 \vect{S}(\grad_{x+y}s_a^1) 
=&\ -\grad_{x+y}[V(s_a^0)]-(\grad_{x+y} V)(s_a^1)\\
 =&\ -(\grad_{x+y}V)(s_a^0)-V(\grad_{x+y}s_a^0)-(\grad_{x+y} V)(s_a^1).
\end{aligned}
\end{equation}
Again, by energy estimate, Proposition \ref{high-deriv-decay}, Lemma \ref{bounded-pot-lem}, Lemma \ref{s0-lem}, and the above bound for $s_a^1$, we see that
\begin{align*}
&\frac{d}{dt} \llp{\grad_{x+y}s_a^1}^2_{L^2_{x, y}}\\
 &\le\ 2\left(\llp{\grad_{x+y}[V(s_a^0)]}_{L^2_{x, y}}+\llp{(\grad_{x+y}V)(s_a^1)}_{L^2_{x, y}} \right)\llp{\grad_{x+y}s_a^1}_{L^2_{x, y}}\\
&\le\   \frac{C}{1+t^3} \llp{\grad_{x+y}s_{2, 0}}_{L^2_{x, y}}\llp{\grad_{x+y}s_a^1}_{L^2_{x, y}}.
\end{align*}
This yields the desired result. Repeat the process to estimate $\grad^j_{x+y}s_a^1$ for $j>1$. 
\end{proof}
\begin{proof}[Sketch of the proof of Proposition \ref{pexc-1}]
We follow closely the proof of Theorem 4.1 in \cite{GM2}.  Write $s_2 = s_a+s_e$ where $s_a$ solves \eqref{sa-eq}. Then we see that $s_e$ and $p_2$ solves a less singular system:
\begin{subequations}
\begin{align}
\vect{S}(s_e) =&\ m\circ p_2 +\overline{p_2}\circ m, \label{r-eq1}\\
\vect{W}(\overline{p_2})=&\ m\circ \overline{s_a}- s_a\circ \overline{m}+m\circ \overline{s_e}- s_e\circ \overline{m} \label{r-eq2}.
\end{align}
\end{subequations}
where $s_e(0) = s_{2, 0}$ and $p_2(0) = p_{2, 0}$. 

Let us define
\begin{align}
E(t)^2:= \llp{s_e(t, \cdot)}^2_{L^2_{x, y}}+\llp{p_2(t, \cdot)}^2_{L^2_{x, y}}.
\end{align}
Then, by energy estimate, we see that
\begin{align*}
\frac{d}{dt} E(t)^2 \le&\ \frac{C}{1+t^3}\left(\llp{p_2(t, \cdot)}_{L^2_{x, y}}\llp{s_e(t, \cdot)}_{L^2_{x, y}} +\llp{p_2(t, \cdot)}_{L^2_{x, y}}\right).
\end{align*}
Hence it follows
\begin{align}
\frac{d}{dt} E(t) \le \frac{C_1}{1+t^3}+ C_2\frac{E(t)}{1+t^3}.
\end{align}
Finally, applying Gr\"onwall's inequality yields uniform the estimate
\begin{align*}
E(t) \le CE(0)\exp\left(C_2\int^t_0  \frac{ds}{1+s^3} \right) \le CE(0).
\end{align*}
To estimate $\grad_{x+y}^js_e$ and $\grad_{x+y}^jp_2$, we begin by differentiating \eqref{r-eq1} and \eqref{r-eq2} with respect to $\grad_{x+y}$ then repeat the above argument.

Finally, let us deduce \eqref{L2Linfty} from \eqref{s2-est}. Observe we have that
\begin{align*}
\Lp{\llp{s_2(x, z)}_{L^2_z}}{L^\infty_x} =&\ \Lp{\llp{s_2(x, x+z)}_{L^2_z}}{L^\infty_x} \le\ \Lp{\llp{s_2(x, x+z)}_{L^\infty_x}}{L^2_z}\\
\le&\ C\sum^2_{j=0}\Lp{\llp{\grad_x^j(s_2(x, x+z))}_{L^\infty_x}}{L^2_z}\\
=&\ C\sum^2_{j=0}\Lp{\llp{(\grad_{x+y}^js_2)(x, x+z))}_{L^2_x}}{L^2_z}\\
=&\ C\sum^2_{j=0}\Lp{(\grad_{x+y}^js_2)(x, y)}{L^2_{x, y}}\le C\sum^2_{j=0}\Lp{\grad_{x+y}^js_{2, 0}}{L^2_{x, y}}.
\end{align*}
\end{proof}

\subsection{Proof of Proposition \ref{s2-collapsing}} In this subsection, we adopt ideas from \cite{GM3} and \cite{GM4}.  In addition to the proof of Proposition \ref{s2-collapsing}, Proposition \ref{linear-prop} is the key result of this subsection. Let us begin by stating a few lemmas. 

\begin{lemma}\label{deriv-s0-lem}
Fix $0<\varepsilon\ll \frac{1}{2}$ (in fact, $\varepsilon$ is fixed for the remainder of the section). Let $s_a^0$ be a solution to \eqref{s0-eq}. Then we have the estimates
\begin{subequations}
\begin{align}
&\llp{s_a^0(t)}_{\dot H^2(\rr^6)} \le\ CN^\frac{3\beta}{2}, \label{L2-s0-est21} \\
&\llp{s_a^0(t)}_{\dot H^{\frac{1}{2}-\alpha}(\rr^6)} \le\ C_\varepsilon \label{L2-s0-est22}
\end{align}
for all $t>0$ where $\alpha := \frac{3\varepsilon}{2(1-\varepsilon)}$. 
Interpolating \eqref{L2-s0-est21} with \eqref{L2-s0-est22} yields
\begin{align}
\llp{s_a^0(t)}_{\dot H^{\frac{1}{2}}(\rr^6)} \le&\ C N^\varepsilon\label{L2-s0-est23}\\
\llp{s_a^0(t)}_{\dot H^{1}(\rr^6)} \le&\ C N^{\frac{\beta}{2}(1+2\varepsilon)}. \label{L2-s0-est24}
\end{align}
\end{subequations}
\end{lemma}

\begin{proof} [Proof]
See Theorem 3 of \cite{Kuz2} for a proof of the lemma. 
\end{proof}

\begin{lemma}
Let $s_a$ be a solution to \eqref{sa-eq}. Then we have the estimate
\begin{subequations}
\begin{align}
&\llp{s_a(t)}_{\dot H^{1}(\rr^6)} \le\ C N^{\frac{\beta}{2}(1+2\varepsilon)},  \label{L2-sa-est21} \\
&\llp{s_a(t)}_{\dot H^2(\rr^6)} \le\  CN^\frac{3\beta}{2} \label{L2-sa-est22}
\end{align}
for all $t>0$. Interpolating with \eqref{L2-sa-est21} with \eqref{L2-sa-est} yields
\begin{align}
&\llp{s_a(t)}_{\dot H^{\frac{1}{2}}(\rr^6)} \le\ C N^{\frac{\beta}{4}(1+2\varepsilon)}.  \label{L2-sa-est23}
\end{align}
\end{subequations}
\end{lemma}

\begin{proof}
Write $s_a = s_a^1+s_a^0$ as in Lemma \ref{sa-lem}. Then we see that $s_a^1$ solves
\begin{align}
\vect{S}(\grad_x s_a^1) = -V(\grad_x s_a^0)-(\grad_xV)(s_a^0)-(\grad_x V)(s_a^1).
\end{align}
Using energy estimate, Lemma \ref{bounded-pot-lem}, \eqref{L2-s0-est}, \eqref{L2-sa-est},  and \eqref{L2-s0-est24}, we see that
\begin{align}
\frac{d}{dt}\llp{\grad_x s_a^1}_{L^2_{x, y}}^2 \le \frac{C N^{\frac{\beta}{2}(1+2\varepsilon)}}{1+t^3} \llp{\grad_x s_a^1}_{L^2_{x, y}}.
\end{align}
Then the result follows after integrating in time. The proof of \eqref{L2-sa-est22} is similar. 
\end{proof}

\begin{prop}\label{linear-prop}
Let $s$ be a solution to the inhomogeneous equation
\begin{align}\label{sg-eq}
\left( \frac{1}{i}\frac{\bd}{\bd t} - \lapl_x-\lapl_y\right)s(t, x, y)   = F
\end{align}
for some forcing $F$ that implicitly depends on time (i.e. solutions to \eqref{sg-eq} are invariant under time-translation). Then  we have the linear estimate
\begin{equation}\label{linear-est}
\begin{aligned}
&\sup_z\llp{s(t, x, x+z)}_{L^2_t([T_0, T_1])L^2_x}\\
&\lesssim_{\varepsilon} \llp{ |\grad_{x-y}|^{\frac{1}{2}} s(T_0)}_{L^2(dxdy)} \\
&\quad\ +(T_1-T_0)\llp{
 | \grad_{x-y}|^{\frac{1}{2}} F}_{N^0([T_0, T_1])} + \llp{| \grad_{x-y}|^{\frac{1}{2}} F}_{L^{2-}_t([T_0, T_1])L^{\frac{6}{5}+}_x L^2_y}
 \end{aligned}
 \end{equation}
 on the interval $[T_0, T_1]$. Here, we used the notation $ \frac{6}{5}+ := \frac{6}{5-4\varepsilon}$ and $2-:=\frac{2}{1+\varepsilon}$. In fact, we also have
\begin{align}\label{linear-est3}
&\sup_z\llp{s(t, x, x+z)}_{L^2_t([T_0, T_1])L^2_x}\\
&\lesssim_{\varepsilon} \llp{ |\grad_{x-y}|^{\frac{1}{2}} s(T_0)}_{L^2(dxdy)} +(T_1-T_0)\llp{
 | \grad_{x-y}|^{\frac{1}{2}} F}_{L^2_t([T_0, T_1])L^\frac{6}{5}_{x-y}L^2_{x+y}} \nonumber\\
 &\quad\ + \llp{| \grad_{x-y}|^{\frac{1}{2}} F}_{L^{2-}_t([T_0, T_1])L^{\frac{6}{5}+}_{x-y} L^2_{x+y}}.\nonumber
 \end{align}
 \end{prop}
 
 \begin{proof}
Let us solve \eqref{sg-eq} on the interval $[T_0, T_1]$ by considering
\begin{align}\label{sg-eq2}
\left(\frac{1}{i}\frac{\bd}{\bd t}-\lapl_x-\lapl_y \right)s = c(t) F , \ \ s(T_0) = s_{0}.
\end{align}
where $c(t) = \chi_{[T_0, T_1]}$.
The solution to \eqref{sg-eq2} is given by
\begin{equation}
\begin{aligned}
s_T =&\  e^{i(t-T_0)\lapl}s_0 -\frac{1}{i}\int^t_{T_0} ds\ e^{i(t-T_0-s)\lapl}c(s)F(s) \\
=&\ e^{i(t-T_0)\lapl}s_0 -\frac{1}{i}\int^\infty_{-\infty} ds\ c(t-s)e^{i(t-T_0-s)\lapl}c(s)F(s) =: s_H+s_P.
\end{aligned}
\end{equation}
Note that $s_T(t) = s(t)$ on the time interval $T_0\le t\le T_1$ (i.e. $c(t)s_T(t) = c(t)s(t)$). Hence it suffices to 
estimate $\sup_z\llp{s_T(t, x, x+z)}_{L^2_tL^2_x}$ and restrict to the interval $[T_0, T_1]$ to obtain the desired estimate.

The estimate for $s_H$ follows from the homogeneous linear estimate in Lemma 5.1 of \cite{GM3}, that is: if 
$f \in L^2(\rr^6)$ then
\begin{align}\label{linear-collapsing}
\llp{(e^{it\lapl}f)(t, x, x)}_{L^2_tL^2_x} \le C\llp{|\grad_{x-y}|^\frac{1}{2}f}_{L^2_{x, y}}.
\end{align}

Next, by a direct computation, we see that
\begin{subequations}
\begin{align}
\widetilde{s_P} =&\ i\hat{c}(\tau+|\xi|^2+|\eta|^2)\widetilde{cF}(\tau, \xi, \eta)\\
=&\ \chi_{\{|\tau+|\xi|^2+|\eta|^2|> 1\}}\widetilde{s_P}+\chi_{\{|\tau+|\xi|^2+|\eta|^2|\le 1\}}\widetilde{s_P}=: \widetilde{s_1}+\widetilde {s_2} \label{split}
\end{align}
\end{subequations}
where
\begin{align}
|\hat c(\tau+|\xi|^2+|\eta|^2)| \lesssim\  \left| \frac{\sin\left( \frac{T_1-T_0}{2}(\tau+|\xi|^2+|\eta|^2)\right)}{\tau+|\xi|^2+|\eta|^2}\right|.
\end{align}
Note that $|\hat c(\tau+|\xi|^2+|\eta|^2)|\lesssim T_1-T_0$ if $|\tau+|\xi|^2+|\eta|^2|\ll 1$ and $|\hat c(\tau+|\xi|^2+|\eta|^2)| \lesssim \langle\tau+|\xi|^2+|\eta|^2\rangle^{-1}$ 
when $|\tau+|\xi|^2+|\eta|^2|\ge 1$, which motivates the splitting \eqref{split}.

By \eqref{linear-collapsing} and a standard fact of dispersive PDE theory (see Lemma 2.9 in \cite{Tao}), we have, for any $\delta>0$, the estimate
\begin{equation}
\begin{aligned}
&\llp{s_1(t, x, x)}_{L^2_t L^2_x}\\
&\lesssim_\delta \llp{|\xi-\eta|^\frac{1}{2}\langle \tau+|\xi|^2+|\eta|^2\rangle^{\frac{1}{2}+\delta} \chi_{\{|\tau+|\xi|^2+|\eta|^2|>1\}}\widetilde{s_P}}_{L^2_\tau L^2_{\xi, \eta}}\\
&\lesssim_\delta \llp{|\xi-\eta|^\frac{1}{2}\langle \tau+|\xi|^2+|\eta|^2\rangle^{-\frac{1}{2}+\delta} \widetilde{cF}}_{L^2_\tau L^2_{\xi, \eta}}.
\end{aligned}
\end{equation}
Finally, by Lemma 4.2 in \cite{GM4}, we can choose $\delta= \delta(\varepsilon)$ (in fact, $\delta = \frac{\varepsilon}{2}$) so that we have
\begin{align}\label{linear-term0}
\llp{s_1(t, x, x)}_{L^2_t([T_0, T_1])  L^2_x} \lesssim_\varepsilon \llp{|\grad_{x-y}|^\frac{1}{2}F}_{L^{2-}_t([T_0, T_1])L^{\frac{6}{5}+}_xL^2_y}.
\end{align}

For $s_2$, we have that
\begin{equation}
\begin{aligned}
&\llp{s_2(t, x, x)}_{L^2_t([T_0, T_1])  L^2_x}\\
&\lesssim_\delta (T_1-T_0)\llp{|\xi-\eta|^\frac{1}{2}\langle \tau+|\xi|^2+|\eta|^2\rangle^{\frac{1}{2}+\delta} \chi_{\{|\tau+|\xi|^2+|\eta|^2|\le 1\}}\widetilde{cF}}_{L^2_\tau L^2_{\xi, \eta}}\\
&\lesssim_\delta (T_1-T_0)\llp{|\xi-\eta|^\frac{1}{2}\langle \tau+|\xi|^2+|\eta|^2\rangle^{-\frac{1}{2}-\delta} \widetilde{cF}}_{L^2_\tau L^2_{\xi, \eta}}.
\end{aligned}
\end{equation}
Then by Lemma 4.1 of \cite{GM4}, we have that
\begin{align}
\llp{s_2(t, x, x)}_{L^2_t([T_0, T_1]) L^2_x} \lesssim_\varepsilon (T_1-T_0)\llp{|\grad_{x-y}|^\frac{1}{2}F}_{N^0([T_0, T_1])}.
\end{align}
Note that \eqref{linear-est3} follows from the parallelogram law,  $2|\xi|^2+2|\eta|^2= |\xi-\eta|^2+|\xi+\eta|^2$
\end{proof}

\begin{cor}
Let $s_{a}$ be a solution to \eqref{sa-eq}.Then we have the estimate
\begin{subequations}
\begin{equation}\label{sa-collapsing-est}
\begin{aligned}
&\sup_z\llp{s_a(t, x, x+z)}_{L^2_t([0, T])L^2_x} \lesssim_\varepsilon\ N^\beta+N^{\beta (1+\varepsilon)}T^{-\frac{2-\varepsilon}{2+2\varepsilon}}+N^{\frac{\beta}{4}(1+2\varepsilon)} \sqrt{T}.
\end{aligned}
\end{equation}
For simplicity, we use the slightly weaker estimate 
\begin{equation}
\begin{aligned}
&\sup_z\llp{s_a(t, x, x+z)}_{L^2_t([0, T])L^2_x} \lesssim_\varepsilon\ N^{\beta (1+\varepsilon)}+N^{\frac{\beta}{4}(1+2\varepsilon)} \sqrt{T}.
\end{aligned}
\end{equation}
\end{subequations}
\end{cor}

\begin{proof} 
 Split $[0, T]$ into approximately $\sqrt{T}$ number of intervals of length $\sqrt{T}$. Let $[T_k, T_{k+1}]=[k\sqrt{T}, (k+1)\sqrt{T}]$ be one of those intervals, then, by Proposition \ref{linear-prop}, we have the estimate
\begin{subequations}
\begin{align}
&\sup_z\llp{s_a(t, x, x+z)}_{L^2_t([T_k, T_{k+1}])L^2_x}\nonumber
\\
& \lesssim\ \llp{|\grad_{x-y}|^{\frac{1}{2}} s_a^1(T_k)}_{L^2_{x, y}}\\
&\quad\ + \sqrt{T}\llp{|\grad_{x-y}|^{\frac{1}{2}} (V(s_a))}_{L^1_t([T_k, T_{k+1}])L^2_{x, y}} \label{linear-term1}\\
&\quad\ + \sqrt{T}\llp{|\grad_{x-y}|^\frac{1}{2}(v_N(x-y)\phi(x)\phi(y))}_{L^2_t([T_k, T_{k+1}])L^\frac{6}{5}_{x-y}L^2_{x+y}} \label{linear-term3}\\
 &\quad\ + \llp{|\grad_{x-y}|^{\frac{1}{2}} (V(s_a))}_{L^{2-}_t([T_k, T_{k+1}]) L^{\frac{6}{5}+}_{x}L^2_y}  \label{linear-term2}\\
  &\quad\ + \llp{|\grad_{x-y}|^{\frac{1}{2}} (v_N(x-y)\phi(x)\phi(y))}_{L^{2-}_t([T_k, T_{k+1}]) L^{\frac{6}{5}+}_{x-y}L^2_{x+y}} \label{linear-term4}
\end{align}
\end{subequations}
For \eqref{linear-term1}, we apply Lemma \ref{bounded-pot-lem}, \eqref{L2-sa-est}, and \eqref{L2-sa-est23} to get the estimate
\begin{align*}
\eqref{linear-term1}&\lesssim \sqrt{T}\llp{\langle\grad_{x-y}\rangle^{\frac{1}{2}} V}_{L^1_t([T_k, T_{k+1}])L^\infty_{x, y}}\llp{\langle\grad_{x-y}\rangle^{\frac{1}{2}}s_a}_{L^\infty_tL^2_{x, y}}\\
&\lesssim \sqrt{T}N^{\frac{\beta}{4}(1+2\varepsilon)} \int^{T_{k+1}}_{T_k} \frac{ds}{s^3} \lesssim \frac{N^{\frac{\beta}{4}(1+2\varepsilon)}}{\sqrt{T}k^2(k+1)}
\end{align*}
For \eqref{linear-term2}, it suffices to consider $|\grad_{x-y}|^\frac{1}{2}((v_N\ast|\phi|^2) s_a)$ and $|\grad_{x-y}|^\frac{1}{2}((v_N\bar\phi\otimes\phi )\circ s_a)$.  Note
that by H\"older, Sobolev, Young's inequalities, \eqref{L2-sa-est}, and  \eqref{L2-sa-est23},  we have that
\begin{align*}
 &\llp{|\grad_{x-y}|^{\frac{1}{2}} ((v_N\ast|\phi|^2) s_a)}_{L^{2-}_t([T_0, T_1]) L^{\frac{6}{5}+}_{x}L^2_y}\\
 &\lesssim \llp{v_N \ast |\grad_x|^\frac{1}{2}|\phi|^2}_{L^{2-}_t([T_0, T_1])L^{3+}_x} \llp{s_a}_{L^\infty_tL^2_{x, y}}\\
 &\quad\ + \llp{v_N \ast |\phi|^2}_{L^{2-}_t([T_0, T_1])L^{3+}_x} \llp{|\grad_{x-y}|^\frac{1}{2}s_a}_{L^\infty_tL^2_{x, y}}\\
  &\lesssim \llp{v_N \ast |\grad_x|^{1+2\varepsilon}|\phi|^2}_{L^{2-}_t([T_0, T_1])L^{2}_x} \llp{s_a}_{L^\infty_tL^2_{x, y}}\\
 &\quad\ + \llp{v_N \ast  |\grad_x|^{\frac{1}{2}+2\varepsilon}|\phi|^2}_{L^{2-}_t([T_0, T_1])L^{2}_x} \llp{|\grad_{x-y}|^\frac{1}{2}s_a}_{L^\infty_tL^2_{x, y}}\\
 &\lesssim \llp{\phi}_{L^{2-}_t([T_0, T_1])L_x^\infty}\llp{\langle\grad_x\rangle^{1+2\varepsilon}\phi}_{L^\infty_tL^2_x} \llp{\langle\grad_{x-y}\rangle^\frac{1}{2}s_a}_{L^\infty_tL^2_{x, y}}\lesssim N^{\frac{\beta}{4}(1+2\varepsilon)}
\end{align*}
The other term is handle in a similar manner, that is
\begin{align*}
 &\llp{|\grad_{x-y}|^{\frac{1}{2}} ((v_N\bar\phi\otimes \phi)\circ s_a)}_{L^{2-}_t([T_0, T_1]) L^{\frac{6}{5}+}_{x}L^2_y}\\
 &\lesssim \int dz\ |v_N(z)|\llp{|\grad_{x-y}|^\frac{1}{2}(\bar \phi(x)\phi(x-z)s_a(x-z, y))}_{L^{2-}_t([T_0, T_1]) L^{\frac{6}{5}+}_{x}L^2_y}\\
  &\lesssim \llp{\phi}_{L^{2-}_t([T_0, T_1])L_x^\infty}\llp{\langle\grad_x\rangle^{1+2\varepsilon}\phi}_{L^\infty_tL^2_x} \llp{\langle\grad_{x-y}\rangle^\frac{1}{2}s_a}_{L^\infty_tL^2_{x, y}}\lesssim N^{\frac{\beta}{4}(1+2\varepsilon)}.
\end{align*}
Lastly, it suffices to treat \eqref{linear-term3} since \eqref{linear-term4} can be handled in exactly the same manner. For \eqref{linear-term3}, the worst scenario occurs when 
$|\grad_{x-y}|^\frac{1}{2}$ lands on $v_N(x-y)$ which then contributes a factor $N^\frac{\beta}{2}$. There we have that
\begin{align*}
&\sqrt{T}\llp{(|\grad_x|^\frac{1}{2}v)_N(x-y) \phi(x)\phi(y)}_{L^2_t([T_0, T_1])L^\frac{6}{5}_{x-y}L^2_{x+y}} \\
&\lesssim \sqrt{T}N^\frac{\beta}{2}\llp{(|\grad|^\frac{1}{2}v)_N}_{L^\frac{6}{5}_x}\llp{\phi}_{L^2([T_k, T_{k+1}])L^\infty_x}\llp{\phi}_{L^\infty_tL^2_x} \lesssim \frac{N^\beta}{k\sqrt{k+1}}.
\end{align*}
Summing the intervals yields the desired result.
\end{proof}

We are now ready to prove Proposition \ref{s2-collapsing}.
\begin{proof}[Proof of Proposition \ref{s2-collapsing}]
Again, write $s_2 = s_a+s_e$ where $s_a$ solves \eqref{sa-eq}. Then we see that $\grad_xs_e$ and $\grad_xp_2$ solves linear system:
\begin{subequations}
\begin{align}
\vect{S}(\grad_xs_e) =&\ -(\grad_x V)(s_e) + v_N\Pi(\phi\otimes \phi)\circ p_2 + m\circ \grad_xp_2\\
&\ +\overline{\grad_x p_2}\circ m+\overline{p_2}\circ v_N\Pi(\phi\otimes \phi)  \nonumber\\
\vect{W}(\grad_x\overline{p_2})=&\ -(v_N\ast \grad_x|\phi|^2)\circ p_2-[v_N\Pi(\bar\phi\otimes \phi), p_2]\\
&\ +v_N\Pi(\phi\otimes \phi)\circ \overline{s_a} -s_a\circ v_N\Pi(\bar\phi\otimes \bar\phi)\nonumber\\
&\ +v_N\Pi(\phi\otimes \phi)\circ \overline{s_e} -s_e\circ v_N\Pi(\bar\phi\otimes \bar\phi)\nonumber\\
&\  +m\circ \overline{\grad_x s_a}- \grad_xs_a\circ \overline{m}+ m\circ \overline{\grad_x s_e}- \grad_xs_e\circ \overline{m}\nonumber
\end{align}
\end{subequations}
Let us define
\begin{align}
E_1(t)^2:= \llp{\grad_xs_e(t, \cdot)}^2_{L^2_{x, y}}+\llp{\grad_xp_2(t, \cdot)}^2_{L^2_{x, y}}.
\end{align}
Then, by energy estimate, we see that
\begin{align}
\frac{d}{dt} E_1(t)^2 \le&\ C \left(\llp{\vect{S}(\grad_x s_e)}_{L^2_{x, y}}\llp{\grad_x s_e}_{L^2_{x, y}}+\llp{\vect{W}(\grad_x p_2)}_{L^2_{x, y}}\llp{\grad_x p_2}_{L^2_{x, y}}\right).
\end{align}
Applying Lemma \ref{bounded-pot-lem}, Proposition  \ref{pexc-1}, \eqref{L2-sa-est}, \eqref{L2-sa-est21}, we have the estimate 
\begin{align*}
&\frac{d}{dt} E_1(t)^2 \\
&\le\ \frac{C}{1+t^3}\left( \llp{\grad_x p_2}_{L^2_{x, y}}\llp{\grad_x s_e}_{L^2_{x, y}}+\llp{\grad_x s_e}_{L^2_{x, y}}+N^\frac{3\beta}{2}\llp{\grad_x p_2}_{L^2_{x, y}}\right)\\
&\le\ \frac{C}{1+t^3}\left(E_1(t)^2+N^{\frac{\beta}{2}(1+2\varepsilon)}E_1(t) \right)
\end{align*}
Hence it follows
\begin{align}
\frac{d}{dt} E_1(t) \le \frac{C}{1+t^3}(E_1(t)+N^{\frac{\beta}{2}(1+2\varepsilon)})
\end{align}
then, by Gr\"onwall's inequality, we have that
\begin{align*}
E_1(t) \le&\ C(E_1(0)+N^{\frac{\beta}{2}(1+2\varepsilon)})\exp\left(C\int^\infty_0 \frac{ds}{1+s^3}\right)\\
 \leq&\ CE_1(0)+CN^{\frac{\beta}{2}(1+2\varepsilon)}.
\end{align*}
The proof of \eqref{s2-collapsing-est} is similar to the proof of \eqref{sa-collapsing-est}.
\end{proof}

\section{Proof of Theorem \ref{main1}}
\subsection{Fock Space Estimates}  We adapt the method of \cite{GM3, Kuz2}
in handling the Fock space error
\begin{subequations}
\begin{align}\label{fock-error}
\Lp{\psi_\text{exact}(t)-\psi_\text{approx}(t)}{\mathcal{F}}.
\end{align}
Since \cite{Kuz2} has provided a detailed account of the strategy for bounding \eqref{fock-error}, we will only give a sketch
of the process.  Nevertheless, for completeness, we have included an appendix with the relevant details. 

By properties of unitary operators, the Fock space error can be rewritten in the form
\begin{align}
\eqref{fock-error} = \Lp{e^{-i\int^t_0 ds X_0(s)}\psi_\text{red}(t)-\Omega}{\mathcal{F}}
\end{align}
\end{subequations}
where $X_0$ is chosen below and $\psi_\text{red}(t)$ is the reduced dynamics defined by
\begin{subequations}
\begin{align}
\psi_\text{red}(t) =&\ e^{\mathcal{B}(t)}e^{\sqrt{N}\mathcal{A}(t)} e^{it\mathcal{H}} e^{-\sqrt{N}\mathcal{A}_0}\Omega, 
\\ \tilde\psi_\text{red}(t) =&\ e^{-i\int^t_0 ds X_0(s)} \psi_\text{red}(t), \ \ E(t) = \tilde \psi_\text{red}(t)-\Omega.
\end{align}
\end{subequations}
A direct computation shows that the error $E$ solves the Cauchy problem
\begin{align}\label{error-eq}
\ \left(\frac{1}{i}\frac{\bd}{\bd t}-\mathcal{H}_\text{red}+X_0\right) E
= \mathcal{H}_\text{red}\Omega-X_0\Omega =:\tilde X, \ E(0, \cdot) = 0
\end{align}
where $\mathcal{H}_\text{red}$ is the reduced Hamiltonian defined by
\begin{subequations}
\begin{align}
\mathcal{H}_\text{red} =&\ N\mu_0(t) + \int dxdy\ \{L(t, x, y)a_x^\ast a_y\}- N^{-\frac{1}{2}}\mathcal{E}(t) \label{red-ham}\\
L(t, x, y)=& -g(t, x, y)+\frac{1}{2}\left( (\bar c_1)^{-1}\circ m\circ \bar s_1+ s_1\circ \bar m\circ (\bar c_1)^{-1}\right)\\
&\ + \frac{1}{2}[\vect{W}(\bar c_1), (\bar c_1)^{-1}]\nonumber\\
\mathcal{E}(t)=&\ e^\mathcal{B}([\mathcal{A}, \mathcal{V}]+N^{-\frac{1}{2}}\mathcal{V})e^{-\mathcal{B}} \ \ \text{ 4th degree polynomial in } (a^\ast, a). \label{error-terms}
\end{align}
\end{subequations}
The complete derivation of \eqref{red-ham} can be found in \cite{GM2} and the explicit form of \eqref{error-terms} is given in the appendix. It has been shown in \cite{GM2} that $\mathcal{H}_\text{red}$ is a fourth degree polynomial in $(a^\ast, a)$ and
its action on $\Omega$ yields
\begin{align}
\mathcal{H}_\text{red}\Omega = (X_0, X_1, X_2, X_3, X_4, 0, 0, \ldots).
\end{align}
In particular, the phase is chosen so that $\tilde X = (0, X_1, X_2, X_3, X_4, 0, 0,\ldots)$. Note that it is safe for us to ignore  the sum of $N\mu_0$ and
the zeroth order term of $N^{-\frac{1}{2}}\mathcal{E}(t)$, which we called $X_0$, in our studies of  \eqref{error-eq}. See appendix for the explicit form of $\tilde X$.

\subsection{Estimates for the Error Terms} Let us rewrite \eqref{error-eq} 
\begin{align}
\vect{S}_FE= (\vect{S}_D-\mathcal{P})E = \tilde X, \ \ E(0, \cdot) = 0
\end{align}
where 
\begin{align}
\vect{S}_D :=&\ \frac{1}{i}\frac{\bd}{\bd t}-\mathcal{H} \ \text{ and }\ \mathcal{P}:=\ \mathcal{H}_\text{red}-\mathcal{H}-X_0.
\end{align}
If we apply energy method directly to estimate $\llp{E}_\mathcal{F}$, then we will end up putting $X_i$ in $L^2(\rr^{3i})$. Hence the best we can do is to put $v_N \in L^2$.
For $X_2$  (modulo all the other difficulties), we have that $\llp{X_2}_{L^2} \sim N^\frac{5\beta-2+\delta}{2}$ ($0<\delta\ll \frac{1}{2}$)
which yields a meaningful result provided $0<\beta<\frac{2-\delta}{5}$. Moreover, we also see that $\llp{X_3}_{L^2}\sim N^\frac{3\beta-1+\delta}{2}$,  
which is only meaningful when $0<\beta<\frac{1-\delta}{3}$. To obtain the result for $0<\beta<\frac{1}{2}$ as in \cite{Kuz2}, we will use Strichartz estimates as in \cite{GM3, Kuz2}.

Define the Strichartz norm on the $n$th sector of $\mathcal{F}$ by
\begin{align*}
\llp{u}_S =&\ \max\Big\{ \llp{u}_{L^\infty(dt)L^2(dx_1\cdots dx_n)}, \llp{u}_{L^2(dt)L^6(d(x_1-x_2))L^2(d(x_1+x_2)\cdots dx_n)}\\
 & \text{ and all other permutations of the variables}\Big\}
\end{align*}
and the dual Strichartz norm
\begin{align*}
\llp{u}_{S'} =&\ \min\Big\{ \llp{u}_{L^1(dt)L^2(dx_1\cdots dx_n)}, \llp{u}_{L^2(dt)L^\frac{6}{5}(d(x_1-x_2))L^2(d(x_1+x_2)\cdots dx_n)}\\
 & \text{ and all other permutations of the variables}\Big\}.
\end{align*}
Let $X \in \mathcal{F}$ be a Fock vector such that all but finitely many of the components are zeros, say $X_0, X_1, \ldots, X_k$ are nonzero. Then 
we define the Strichartz norm 
\begin{subequations}
\begin{align}
\llp{X}_S = \max\left\{ |X_0|, \llp{X_1}_S, \ldots, \llp{X_k}_S\right\}
\end{align}
and similarly for the dual Strichartz norm
\begin{align}
\llp{X}_{S'} = \max\left\{ |X_0|, \llp{X_1}_{S'}, \ldots, \llp{X_k}_{S'}\right\}.
\end{align}
\end{subequations}
We will also denote the Strichartz norm on the interval $[0, T]$ by $\llp{X}_{S_T}$.

By standard arguments, we have the following lemma; see Lemma 9.2 in \cite{GM3} for the proof.

\begin{lemma}\label{fock-strichartz}
Let $f$ be a Fock vector with zero entries past the $k$th sector. Assume $\psi$ is a solution to 
\begin{align}
\vect{S}_D \psi = f, \ \psi(0, \cdot) = 0.
\end{align}
Then we have the estimate
\begin{align}
\llp{\psi}_S \lesssim \llp{f}_{S'}.
\end{align}
By definition, it follows that $\sup_t \llp{\psi}_\mathcal{F} \lesssim \llp{\psi}_S$.
\end{lemma}

Following the argument in \cite{Kuz2}, we define
\begin{subequations}
\begin{align}
X_2^s(y_1, y_2) =&\ \frac{1}{2N} v_N(y_1-y_2)\{s_1(y_1, y_2)+(\bar p_1\circ s_1)(y_1, y_2)\}\\
=&\ \frac{1}{4N}v_N(y_1-y_2)s_2(y_1, y_2)\nonumber\\
X_3^s(y_1, y_2, y_3)=&\  \frac{1}{\sqrt{N}}v_N(y_1-y_2)\phi(y_2)s_1(y_1, y_3)\\
X^r_i =&\
\begin{cases}
 X_i-X^s_i& \text{ for } i =2, 3,\\
 X_i & \text{ for } i=1, 4
 \end{cases}.
\end{align}
\end{subequations}
Let us split $E=E^r+E^s$ where $E^r$ and $E^s$ solve the Cauchy problems
\begin{subequations}
\begin{align}
\vect{S}_F E^r=&\ X^r := (0, X_1, X_2^r, X_3^r, X_4, 0, \cdots),\ \ E^r(0, \cdot) = 0\label{regular-error}\\
\vect{S}_F E^s=&\ X^s := (0, 0, X_2^s, X_3^s, 0, 0, \cdots),\ \ E^s(0, \cdot) = 0. \label{singular-error}
\end{align}
\end{subequations}
The superscript $r$ in $E^r$ indicates the part of $E$ that corresponds to forcing terms of $\tilde X$ that are ``regular", whereas, the superscript $s$ 
refers to the part of $E$ that corresponds to the more ``singular" forcing terms. 

\subsubsection{Estimates for $E^r$}
We can readily estimate $E^r$ by energy method. By energy estimate, we have that
\begin{align}
\llp{E^r(t)}_\mathcal{F} \lesssim \sum^4_{i=1}\int^t_0 d\tau\ \llp{X^r_i(\tau)}_{L^2}.
\end{align}
Hence it suffices to estimate the $L^2$-norm of $X^r_i$. 

\begin{prop}
We have the following estimates
\begin{subequations}\label{Xr-est}
\begin{align}
\int^T_0 d\tau\ \llp{X_1(\tau)}_{L^2} \lesssim&\ N^{-\frac{1}{2}}C_0(T, N),  \label{X_1^r-est}\\
\int^T_0 d\tau\ \llp{X_2^r(\tau)}_{L^2} \lesssim&\ N^{-1} C_0(T, N),  \label{X_2^r-est}\\
\int^T_0 d\tau\ \llp{X_3^r(\tau)}_{L^2} \lesssim&\ N^{-\frac{1}{2}}, \label{X_3^r-est}\\
\int^T_0 d\tau\ \llp{X_4(\tau)}_{L^2} \lesssim&\ N^{\frac{3\beta-2}{2}}T. \label{X_4^r-est}
\end{align}
\end{subequations}
In particular,  for any $0<\beta<\frac{1}{2}$ and interval $[0, T]$, we have that
\begin{align}
\llp{E^r(t)}_\mathcal{F} \lesssim  N^{-\frac{1}{2}}C_0(T, N)+N^{\frac{3\beta-2}{2}}T \lesssim N^{-\frac{1}{2}+\beta(1+\varepsilon)}.
\end{align}
for all $t \in [0, T]$ and $N$ sufficiently large.

\end{prop}

\begin{proof}
The proof of \eqref{Xr-est} follows immediately from Lemma \ref{p1-lem}, \ref{p2r-lem}, \ref{p3r-lem}, and \ref{p4-lem}.
\end{proof}

\subsubsection{Estimates for $E^s$} In the spirit of Section 4, we split $E^s = E^s_a+E^s_e$ where 
\begin{subequations}\label{fock-schrod}
\begin{align}
\vect{S}_D E^s_a =&\ X^s, \ \ E^s_a(0, \cdot) = 0, \label{fock-schrod1}\\
(\vect{S}_D-\mathcal{P}) E^s_e =&\ \mathcal{P}E^s_a, \ \ E^s_e(0, \cdot) = 0. \label{fock-schrod2}
\end{align}
\end{subequations}

\begin{prop}
Let $0<\beta<\frac{1}{2}$ and $E^s$ solves \eqref{fock-schrod}. Then we have that
\begin{align}
\llp{E^s(t)}_\mathcal{F} \lesssim_\varepsilon N^\frac{3\beta-2+2\beta\varepsilon}{2}.
\end{align}
for all $t \in [0, T]$ and $N$ sufficiently large.
\end{prop}

\begin{proof}
Let us first fix an interval $[0, T]$. The proof of the proposition is based on the proof of Theorem 9.3 in \cite{GM3} via iteration method.

Using Lemma \ref{fock-strichartz}, \ref{p2s-lem}, and  \ref{p3s-lem}, we immediately see that
\begin{align}
\llp{E_a^s(t)}_\mathcal{F} \lesssim \llp{E_a^s}_{S_t} \lesssim \llp{X^s}_{S'_t} \lesssim N^{\frac{\beta-2}{2}}C_0(T, N).
\end{align}
By the appendix, it is helpful to split $\mathcal{P}$ into $\mathcal{P}^s$ and $\mathcal{P}^r$ where $\mathcal{P}^s: = \mathcal{P}_2^s+\mathcal{P}_3^s$. Note, by Lemma \ref{p2s-lem} and \ref{p3s-lem}, we have the estimates
\begin{align}
\llp{\mathcal{P}^s E}_{S'_t} \lesssim N^{\frac{\beta-2}{2}}C_0(T, N)\llp{E}_{S_t}.
\end{align}
For notational convenience, let $E_1=E_a^s$ and consider $E_2= E_1+\vect{S}_D^{-1}\mathcal{P}^sE_1$ (first iteration). Note that $E_1$ has at most $5$ entries since $\vect{S}_D$ is diagonal and $X^s$ is zero after the first 5 sectors; likewise, $E_2$ has at most 9 nonzero entries.  Then, by Strichartz estimates, we have that
\begin{equation}\label{E2-strich}
\begin{aligned}
\llp{E_2(t)}_\mathcal{F} 
&\lesssim\ \llp{E_1}_{S_t}+\llp{\vect{S}_D^{-1}\mathcal{P}^s E_1}_{S_t}\\
&\lesssim\ N^{\frac{\beta-2}{2}}C_0(T, N)+N^{\beta-2}C_0(T, N)^2\\
&\lesssim\ N^\frac{3\beta-2+2\beta\varepsilon}{2}+N^\frac{3\beta-4+2\beta\varepsilon}{4}\sqrt{T}+N^\frac{3\beta-4+2\beta\varepsilon}{2}T
\end{aligned}
\end{equation}
Finally, note that $E^s-E_2$ solves
\begin{align}
(\vect{S}_D-\mathcal{P})(E^s-E_2)=\mathcal{P}^rE_1 + \mathcal{P}\vect{S}_D^{-1}\mathcal{P}^s E_1.
\end{align}
By energy estimate, we have that
\begin{align}\label{fock-energy}
\llp{E^s(t)-E_2(t)}_\mathcal{F} \lesssim  \int^t_0 d\tau\ \{\llp{\mathcal{P}^rE_1(\tau)}_\mathcal{F} +\llp{\mathcal{P}^s\vect{S}_D^{-1}\mathcal{P}^s E_1(\tau)}_\mathcal{F}\}.
\end{align}
By Lemma \ref{p1-lem}, \ref{p2r-lem}, \ref{p3r-lem}, and \ref{p4-lem}, we have that
\begin{align*}
\text{LHS of }\eqref{fock-energy}  \lesssim&\ (1+N^{\frac{3\beta-2}{2}}T+N^{-\frac{1}{2}}C_0(T, N))N^{\frac{\beta-2}{2}}C_0(T, N)\\
&\ +(N^\frac{3\beta-1}{2}+N^\frac{3\beta-2}{2}\sqrt{T}C_0(T, N))N^{\beta-2}C_0(T, N)^2.
\end{align*}
This completes the proof.
\end{proof}

 \section{Application: Derivation of The Focusing NLS in $\rr^3$}
We provide two derivation of the focusing  NLS. For the first derivation, we employ the method introduced by Pickl in \cite{Pickl, Pickl2}.  In the second
 derivation, we apply the works of Nam and Napi\'orkowski on the $N$-norm approximation  of 
the many-body dynamics developed in \cite{NaNa}.  
 
 \subsection{Pickl's Method}
Following closely the presentation in \cite{Pickl2}, we consider the quantities:
\begin{definition}
Let $\phi \in L^2(\rr^3)$
\begin{enumerate}[(a)]
\item For each $1\leq j \leq N$ we define the projectors $p^\phi_{j}: L^2(\rr^{3N}) \rightarrow L^2(\rr^{3N})$ 
and $q^\phi_{j}: L^2(\rr^{3N}) \rightarrow L^2(\rr^{3N})$ given by
\begin{align}
 p^\phi_j \Psi_N(x_1, \ldots, x_N) = \phi(x_j) \int \phi^\ast(x_j') \Psi_N(x_1, \ldots, x_j', \ldots, x_N) dx_j
\end{align}
and $q^\phi_j= 1-p^\phi_j$ respectively. 
\item Furthermore, for any $1 \leq k \leq N$ we defined $P^\phi_k: L^2(\rr^{3N}) \rightarrow L^2(\rr^{3N})$
given by
\begin{align}
 P^\phi_k : = \sum_{a \in \mathcal{A}_k} \prod^N_{\ell = 1} (p^\phi_\ell)^{1-a_\ell}(q_\ell^\phi)^{a_\ell}
\end{align}
where
\begin{align}
 \mathcal{A}_k = \{ (a_1, \ldots, a_N) \mid a_i \in \{0, 1\} \text{ and } \sum_{i=1}^N a_i = k\}
\end{align}

 \item Assume $0<\lambda \leq 1$. Let us define the function $m^\lambda:\{1, \ldots, N\} \rightarrow \rr_{\geq 0}$
given by
\begin{align}
 m^{\lambda}(k) :=
\begin{cases}
 kN^{-\lambda}, \text{ for } k\leq N^\lambda, \\
1, \text{ otherwise}
\end{cases}
\end{align}
and a corresponding functional $\alpha_N^\lambda: L^2(\rr^{3N}) \times L^2(\rr^3) \rightarrow \rr_{\geq 0}$ given by
\begin{subequations}
\begin{align}
 \alpha_N^\lambda(\Psi_N, \phi):=&\ \inprod{\Psi_N}{\sum^{N}_{k=1}m^\lambda(k)P^\phi_j\Psi_N} \\
 =&\ \inprod{\Psi_N}{\widehat m^{\lambda, \phi}\Psi_N}=\llp{(\widehat m^{\lambda, \phi})^{1/2}\Psi_N}^2_{L^2_\vect{x}}.
\end{align}
\end{subequations}
For convenience, we shall use the notation $\alpha_N$ instead of $\alpha^1_N$. 
\end{enumerate}
\end{definition}

As a direct consequence of the definitions, we obtain the inequality
\begin{align}
 \alpha_N(\Psi_N, \phi) = \llp{q_1^\phi\Psi_N}^2_{L^2_\vect{x}} \leq \alpha^{\lambda}_N(\Psi_N, \phi)
\end{align}
for $0< \lambda < 1$. Again, by the definition, we could derive an error
 bound for the rate of convergence of the one-particle density towards the mean-field limit
\begin{align*}
 \llp{ \gamma_N^{(1)}- |\phi\rangle \langle \phi|}_{\operatorname{op}} \leq&\  
\big|\llp{p_1^\phi\Psi_N}^2_{L^2_\vect{x}}-1\big|\llp{|\phi\rangle\langle \phi|}_{\operatorname{op}}\\
&+ 2\llp{q_1^\phi\Psi_N}_{L^2_\vect{x}}\llp{p_1^\phi\Psi_N}_{L^2_\vect{x}}+\llp{q_1^\phi \Psi_N}^2_{L^2_\vect{x}}\\
\leq&\ \big|\llp{p_1^\phi\Psi_N}^2_{L^2_\vect{x}}-1\big|+
 2\llp{q_1^\phi\Psi_N}_{L^2_\vect{x}}\llp{p_1^\phi\Psi_N}_{L^2_\vect{x}}\\
 &+\llp{q_1^\phi \Psi_N}^2_{L^2_\vect{x}}\\
\lesssim&\  \llp{q_1^\phi \Psi_N}^2_{L^2_\vect{x}}+\llp{q_1^\phi \Psi_N}_{L^2_\vect{x}}.
\end{align*}
Since $|\phi\rangle \langle \phi|$ is a rank one projection operator, by remark 1.4 in \cite{RS} the trace norm is two times the operator norm, i.e.,  $2\llp{\gamma_N^{(1)}- |\phi\rangle \langle \phi|}_{\operatorname{op}} 
= \Tr\left|\gamma_N^{(1)}- |\phi\rangle \langle \phi|\right|$. Then it follows from the above estimates
\begin{align}
\Tr\left|\gamma_{N, t}^{(1)}- |\phi_t\rangle \langle \phi_t| \right| \lesssim \alpha^\lambda_N(\Psi_N, \phi_t)+\sqrt{\alpha^\lambda_N(\Psi_N, \phi_t)}.
\end{align}
Thus, to obtain a rate of convergence for the error it suffices to prove an estimate for $\alpha_N^\lambda(\Psi_N, \phi)$. Let us now state the main theorem in \cite{Pickl2} 
that we will use to derive the focusing NLS:
\begin{theorem}\label{Pickl-thm}
 Assume $0<\lambda, \beta<1$ and $v_N$ satisfies the same conditions as before. Assume for every $N \in \nn$ there
exists a solution to the linear $N$-body Schr\" odinger equation $\Psi_N(t, x)$ and a $L^\infty$ solution of the
mean-field equation $\psi_t$ on some interval $[0, T)$ with $T \in \rr_{>0}\cup\{\infty\}$. Then for any $t \in [0, T)$
\begin{align}
 \alpha_N^{\lambda}(\Psi_{N, t}, \psi_t) \leq&\ \exp\left( \int^t_0 C_v\llp{\phi_s}^2_{L^\infty_x}\ ds \right)\alpha_N^\lambda(\Psi_{N, 0}, \phi_0)\\
&\ + \left[\exp\left(C_v\int^t_0 \llp{\phi_s}_{L^\infty_x}^2\ ds\right)-1\right] \sup_{0\leq s \leq t}K^{\phi_s} N^{\delta_\lambda}\nonumber
\end{align}
where $\delta_\lambda = \frac{1}{2}\max\{1-\lambda-4\beta, 3\beta-\lambda, -1+\lambda +3\beta\}$, $C_v$ is some constant depending only on $v$ and 
\begin{align}
 K^\phi:= C_v( \llp{ \lapl |\phi|^2}_{L^2_x}+\llp{\phi}_{L^\infty_x} +1) \llp{\phi}_{L^\infty_x}.
\end{align}
\end{theorem}

\begin{proof}[Proof of Theorem \ref{DerFNLS} for $0<\beta<\frac{1}{6}$]
 Note, if $\Psi_N(0, x) = \phi^{\otimes N}$ then $\alpha^\lambda_N(\phi^{\otimes N}, \phi)=0$. Hence combining Theorem \ref{Pickl-thm} and our above decay result for $\phi$ satisfying  the focusing NLS equation \eqref{FNLS}, we have that
\begin{align*}
\alpha_N^\lambda(\Psi_{N, t}, \phi_t) \leq \left[\exp\left(C_v\int^t_0 \llp{\phi_s}_{L^\infty_x}^2\ ds\right)-1\right] \sup_{0\leq s \leq t}K^{\phi_s} N^{\delta_\lambda}
\end{align*}
where 
\begin{align*}
K^{\phi_t} =&\  C_v( \llp{ \lapl |\phi_t|^2}_{L^2_x} + \llp{\phi_t}_{L^\infty_x} + 1)\llp{\phi_t}_{L^\infty_x}\\
\lesssim&\ ( \llp{|\grad_x \phi_t|^2}_{L^2_x}+\llp{\phi_t\lapl \bar \phi_t}_{L^2_x} + \llp{\phi_t}_{L^\infty_x} + 1)\llp{\phi_t}_{L^\infty_x}\\
\lesssim&\ (\llp{\grad_x\phi_t}_{L^\infty_x}\llp{\grad_x\phi_t}_{L^2_x}+\llp{\phi_t}_{L^\infty_x} \llp{\grad_x^2\phi_t}_{L^2_x}+\llp{\phi_t}_{L^\infty_x}+1)\llp{\phi_t}_{L^\infty_x}\\
\lesssim&\ \frac{1}{1+t^{\frac{3}{2}}}.
\end{align*}
Thus, it follows
\begin{align*}
\Tr\left|\gamma_{N, t}^{(1)}- |\phi_t\rangle \langle \phi_t| \right| \lesssim 
\sqrt{\alpha^\lambda_N(\Psi_{N, t} , \varphi_t)} \lesssim N^{\delta_\lambda/2}.
\end{align*}
By remark 1 in \cite{Pickl2}, we see there is a choice of $\lambda$ such 
that $\delta_\lambda <0$ when $0<\beta<\frac{1}{6}$.  
\end{proof}

 \subsection{$N$-Norm Approximation Method} By Proposition \ref{decay} and Remark 4 in \cite{NaNa}, we extend the result of Theorem 6 in \cite{NaNa} to 
 the case of attractive interaction. More precisely, we have following proposition.
 
 \begin{prop}\label{N-norm-approx}
Let $v \in C^\infty_c(\rr^3)$ and $v\le 0$. Assume $\phi$ solves \eqref{htnls} with initial conditions $\phi_0 \in L^2(\rr^2)\cap W^{m, 1}(\rr^3)$ and $\llp{\phi_0}_{L^2_x}=1$ 
for some $m$ sufficiently large and  $\dot H^{\frac{1}{2}}_x$-norm sufficiently small, depending on $v$. Let $(\psi_n(t))_{n=0}^\infty = e^{-\mathcal{B}(k_t)}\Omega \in \mathcal{F}$ 
where $k_t$ solves the linear system of equations \eqref{eeq1} and \eqref{eeq2} for some initial $k(0, \cdot)$. Then the $N$-body evolution $\Psi_N(t) = e^{itH_N}\Psi_N(0)$ with the initial state
\begin{align}
\Psi_N(0) = \sum^N_{n=0} \phi_0^{\otimes (N-n)}\otimes_s \psi_n(0)
\end{align}
satisfies the norm approximation
\begin{align}
\Lp{\Psi_N(t)-\sum^N_{n=0} \phi(t)^{\otimes(N-n)}\otimes_s\psi_n(t)}{L^2(\rr^{3N})} \le C_1(t)N^\frac{3\beta-1}{2}
\end{align}
where
\begin{align}
C_1(t) \leq C(1+t)\left(1+\log(1+t)+\langle e^{-\mathcal{B}(k_0)}\Omega, \mathcal{N}e^{-\mathcal{B}(k_0)}\Omega\rangle \right)^4
\end{align}
for some constant $C>0$ depending only on $\phi_0$. 
 \end{prop}

 \begin{proof}[Proof of Theorem \ref{DerFNLS} for $\frac{1}{6}\le \beta<\frac{1}{3}$] 
 Take $k(0, \cdot) =0$. Then the proof follows immediately from Proposition \ref{N-norm-approx} and Corollary 2 in \cite{LNS}.
 \end{proof}

\appendix \section{Estimates for $\mathcal{P}$ }
\subsection{Normal Ordering of $\mathcal{H}_\text{red}$} The complete explicit form of $\mathcal{H}_\text{red}$ has been provided 
in equations (26)-(30) of \cite{Kuz2}, which is based on the computation in Section 5 of \cite{GM2}. The purpose of this section is to give the normal ordering of $\mathcal{H}_\text{red}$, which will be
useful when estimating the Fock space error. 

Recall that the reduced Hamiltonian is a self-adjoint operator given by
\begin{align}
\mathcal{H}_\text{red} =&\ N\mu(t) + \int dxdy\ \{L(t, x, y)a^\ast_x a_y\}-N^{-\frac{1}{2}}\mathcal{E}(t)\\
L(t, x, y) =&\ \lapl_x \delta(x-y)\\
&\ - (v_N\ast|\phi|^2)(t, x)\delta(x-y) -v_N(x-y)\phi(t, x)\bar\phi(t, y) \nonumber\\
&\  +\frac{1}{2}\Big( \bar c_1^{-1}\circ m\circ \bar s_1 + s_1\circ \bar m \circ \bar c_1^{-1}+ \left[ \vect{W}(\bar c_1), \bar c_1^{-1}\right]\Big). \nonumber
\end{align}
where $\mu(t)$ is some scalar function which is irrelevant for the purposes of this paper. 
The first two terms of $\mathcal{H}_\text{red}$ are already trivially in normal ordering. 

Let us now write down the normal ordering of the fourth degree polynomial $\mathcal{E}(t)$.

\subsubsection{Quartic polynomials} Here we use the notations $u=\sh(k) = s_1$ and $c = \ch(k)= c_1$. Let us start with the $a^\ast a^\ast a a$ terms. Following \cite{Kuz2}, we divided them into self-adjoint operators and
non-self-adjoint operators. For the self-adjoint terms we have
\begin{subequations}\label{2-2terms}
\begin{align}
& \frac{1}{2N}\int dx_1dx_2 dy_1dy_2dy_3 dy_4 \Big\{  \nonumber\\
&\quad\ \bar c(y_1, x_1)c(x_2, y_2)v_N(x_1-x_2) c(y_3, x_1)\bar c( x_2, y_4) a^\ast_1 a^\ast_2 a_3 a_4 \label{2-2terms1}\\
&+ \bar c(y_1, x_1)\bar u(x_2, y_2)v_N(x_1-x_2) c(y_3, x_1) u( x_2, y_4) a^\ast_1 a^\ast_4 a_2 a_3 \label{2-2terms2}\\
&+ \bar u(y_1, x_1)c(x_2, y_2)v_N(x_1-x_2) u(y_3, x_1)\bar c( x_2, y_4) a^\ast_2 a^\ast_3 a_1 a_4 \label{2-2terms3}\\
&+ \bar u(y_1, x_1)\bar u(x_2, y_2)v_N(x_1-x_2) u(y_3, x_1) u( x_2, y_4)a^\ast_3a^\ast_4 a_1 a_2 \label{2-2terms4} \Big\}.
\end{align}
\end{subequations}
The non-self-ajoint $a^\ast a^\ast a a$ terms are given by
\begin{subequations}
\begin{align}
&\ \frac{1}{2N}\int dx_1dx_2 dy_1dy_2dy_3 dy_4 \Big\{  \nonumber\\
&\quad \bar u(y_1, x_1)c(x_2, y_2)v_N(x_1-x_2) c(y_3, x_1)u( x_2, y_4) a^\ast_2 a^\ast_4 a_1 a_3\Big\}
\end{align}
\end{subequations}
and its adjoint. 

For the $a^\ast a a a$ terms, we have
\begin{subequations}\label{1-3terms}
\begin{align} 
& \frac{1}{2N}\int dx_1dx_2 dy_1 dy_2 dy_3 dy_4 \Big\{  \nonumber\\
&\quad\ \bar c(y_1, x_1)\bar u(x_2, y_2)v_N(x_1-x_2)c(y_3, x_1)\bar c(x_2, y_4) a^\ast_1 a_2 a_3a_4\\
&+\bar u(y_1, x_1)\bar u(x_2, y_2)v_N(x_1-x_2) c(y_3, x_1) u(x_2, y_4) a^\ast_4 a_1 a_2  a_3 \\
&+\bar u(y_1, x_1) \bar u(x_2, y_2)v_N(x_1-x_2) c(y_3, x_1) \bar c(x_2, y_4) a^\ast_2 a_1a_3a_4 \\
&+ \bar u(y_1, x_1) \bar u(x_2, y_2)v_N(x_1-x_2) c(y_3, x_1) \bar c(x_2, y_4) a^\ast_3 a_1 a_2 a_4\Big\}
\end{align}
\end{subequations}
and the $a a a a $ term is given by
\begin{align}\label{0-4terms}
& \frac{1}{2N}\int dx_1dx_2 dy_1dy_2dy_3 dy_4 \Big\{  \nonumber\\
&\quad\ \bar u(y_1, x_1)\bar u(x_2, y_2)v_N(x_1-x_2) c(y_3, x_1)\bar c(x_2, y_4)  a_1a_2 a_3 a_4 \Big\}.
\end{align}
Taking the adjoint of \eqref{1-3terms} and \eqref{0-4terms} yield the $a^\ast a^\ast a^\ast a$ and $a^\ast a^\ast a^\ast a^\ast$ terms.

\subsubsection{Cubic polynomials} For the $aaa$ terms, we have
\begin{subequations}\label{0-3terms}
\begin{align}
&\frac{1}{\sqrt{N}}\int dx_1dx_2 dy_1 dy_2 dy_3  \Big\{  \nonumber\\
&\quad\  \bar u(y_1, x_1)v_N(x_1-x_2)\bar\phi(x_2)c(y_2, x_1) \bar c(x_2, y_3)\label{0-3term1}\\
&+ \bar u(y_1, x_1)v_N(x_1-x_2)\phi(x_2)\bar u(y_2, x_2) \bar c(x_1, y_3) \label{0-3term2}
 \Big\}a_1a_2a_3
\end{align}
\end{subequations}

For the $a^\ast a a$ terms, we have
\begin{subequations}\label{1-2terms}
\begin{align}
&\frac{1}{\sqrt{N}}\int dx_1dx_2 dy_1 dy_2 dy_3  \Big\{  \nonumber\\
&\quad\ \bar c(y_1, x_1)v_N(x_1-x_2)\bar\phi(x_2)c(y_2, x_1) \bar c(x_2, y_3) a^\ast_1 a_2a_3\label{1-2term1}\\
&+ c(y_1, x_1)v_N(x_1-x_2)\phi(x_2)\bar u(y_2, x_2) \bar c(x_1, y_3) a^\ast_1 a_2a_3 \label{1-2term2}\\
&+ \bar u(y_1, x_1)v_N(x_1-x_2)\phi(x_2)\bar c(y_2, x_2) \bar c(x_1, y_3) a^\ast_2 a_1a_3\\
&+ \bar u(y_1, x_1)v_N(x_1-x_2)\bar\phi(x_2)u(y_2, x_1)\bar c(x_2, y_3) a^\ast_2 a_1 a_3\\
&+ \bar u(y_1, x_1)v_N(x_1-x_2)\bar\phi(x_2)c(y_2, x_1)u(x_2, y_3) a^\ast_3 a_1 a_2\\
&+ \bar u(y_1, x_1)v_N(x_1-x_2)\phi(x_2)\bar u(y_2, x_2)u(x_1, y_3) a^\ast_3 a_1 a_2\label{1-2term6}
\Big\}.
\end{align}
\end{subequations}
Taking the adjoint of \eqref{1-2terms} and \eqref{0-3terms} give us the $a^\ast a^\ast a$ and  $a^\ast a^\ast a^\ast$ terms.

\subsubsection{Quadratic polynomials} Like the  $a^\ast a^\ast a a$ terms, we divide that $a^\ast a$ terms into two groups, self-adjoint 
and non-self-adjoint. For the self-adjoint operators, we have 
\begin{subequations}\label{sa-1-1terms}
\begin{align}
 &\frac{1}{2N}\int dx_1dx_2 dy_1 dy_2 \Big\{  \nonumber\\
 &+  (u\circ \bar u)(x_1, x_1)v_N(x_1-x_2) c (y_1, x_2) c(x_2, y_2)a^\ast_2 a_1 \label{sa-1-1term1}\\
 &+ (u\circ \bar u)(x_2, x_2)v_N(x_1-x_2) \bar c (y_1, x_1) \bar c(x_1, y_2)a^\ast_1 a_2\\
 &+ (u\circ \bar u)(x_2, x_2)v_N(x_1-x_2) u(y_1, x_1) \bar u(x_1, y_2)a^\ast_2 a_1\\
&+ (u\circ \bar u)(x_1, x_1)v_N(x_1-x_2) u(y_1, x_2) \bar u(x_2, y_2)a^\ast_1 a_2\\
&+  (u\circ \bar u)(x_1, x_2)v_N(x_1-x_2)\bar u(y_1, x_1) u(x_2, y_2) a^\ast_2 a_1\\
&+  (u\circ \bar u)(x_2, x_1)v_N(x_1-x_2)u(y_1, x_1) \bar u(x_2, y_2) a^\ast_1 a_2 \label{sa-1-1term6}
 \Big\}.
\end{align}
\end{subequations}
In the case of non-self-adjoint operators, we have the operators
\begin{subequations}\label{1-1terms}
\begin{align}
 &\frac{1}{2N}\int dx_1dx_2 dy_1 dy_2 \Big\{  \nonumber\\
 &\quad\ (\bar u\circ \bar c)(x_2, x_1)v_N(x_1-x_2)c(y_1, x_1)c(x_2, y_2) a^\ast_2 a_1 \label{1-1term1}\\
  &+ (\bar u\circ \bar c)(x_1, x_2)v_N(x_1-x_2)u(y_1, x_1)\bar c(x_2, y_2) a^\ast_1 a_2 \label{1-1term2}\\
 &+ (\bar u\circ \bar c)(x_1, x_2)v_N(x_1-x_2)c(y_1, x_1)u(x_2, y_2) a^\ast_2 a_1
 \Big\}
\end{align}
\end{subequations}
and the adjoints operators. 

For the $aa$ terms, we have
\begin{subequations}\label{0-2terms}
\begin{align}
 &\frac{1}{2N}\int dx_1dx_2 dy_1 dy_2 \Big\{  \nonumber\\
 &\quad\ (\bar u\circ \bar c)(x_1, x_2)v_N(x_1-x_2)c(y_1, x_1)\bar c(x_2, y_2) \label{0-2term1}\\ 
  &+(u\circ  c)(x_1, x_2)v_N(x_1-x_2)\bar u(y_1, x_1)\bar u(x_2, y_2) \label{0-2term2}\\
      &+(u \circ \bar u)(x_1, x_2) v_N(x_1-x_2) \bar u(y_1, x_1) \bar c(x_2, y_2) \\
      &+(\bar u \circ u)(x_1, x_2) v_N(x_1-x_2)  c(y_1, x_1) \bar u(x_2, y_2)  \\
          &+(u\circ \bar u)(x_1, x_1) v_N(x_1-x_2) \bar u(y_1, x_2) \bar c(x_2, y_2) \\
  &+(u \circ \bar u)(x_2, x_2) v_N(x_1-x_2) \bar u(y_1, x_1) \bar c(x_1, y_2) 
    \Big\}a_1 a_2.
\end{align}
\end{subequations}
Taking the adjoint of \eqref{0-2terms} yields all the $a^\ast a^\ast$ terms.

\subsubsection{Linear polynomials} Lastly, the $a$ terms are given by 
\begin{subequations}\label{0-1terms}
\begin{align}
&\frac{1}{\sqrt{N}}\int dx_1dx_2 dy_1\Big\{ \nonumber \\
 &\quad\ (c\circ \bar u)(x_1, x_2)v_N(x_1-x_2)\phi(x_2)c(y_1, x_1) \label{0-1term1}\\
 &+ (u\circ c)(x_1, x_2)v_N(x_1-x_2)\bar \phi(x_2)\bar u(y_1, x_1) \label{0-1term2} \\
 &+ (u\circ \bar u)(x_1, x_1) v_N(x_1-x_2)\bar\phi(x_2) c(y_1, x_2) \label{0-1term3}\\
 &+ (u\circ \bar u)(x_2, x_1) v_N(x_1-x_2)\bar\phi(x_2) c(y_1, x_1) \label{0-1term4}\\
 &+ (u\circ \bar u)(x_1, x_1) v_N(x_1-x_2)\phi(x_2) \bar u(y_1, x_2) \label{0-1term5} \\
 &+ (u\circ \bar u)(x_1, x_2) v_N(x_1-x_2)\phi(x_2) \bar u(y_1, x_1) \Big\}a_1.\label{0-1term6}
\end{align}
\end{subequations}
 Taking the adjoint of \eqref{0-1terms} gives all the $a^\ast$ terms. 

\subsection{Explicit form of $\mathcal{H}_\text{red}\Omega$}  By direct calculation, we see that
\begin{subequations}
\begin{align}
X_1(y_1)=&\frac{1}{\sqrt{N}}\int dx_1dx_2 \Big\{ \nonumber \\
 &\quad\ (\bar c\circ  u)(x_1, x_2)v_N(x_1-x_2)\bar\phi(x_2)\bar c(y_1, x_1) \\
 &+ (\bar u\circ \bar c)(x_1, x_2)v_N(x_1-x_2) \phi(x_2)u(y_1, x_1)  \\
 &+ (u\circ \bar u)(x_1, x_1) v_N(x_1-x_2)\phi(x_2) \bar c(y_1, x_2) \\
 &+ (u\circ \bar u)(x_1, x_2) v_N(x_1-x_2)\phi(x_2) \bar c(y_1, x_1) \\
 &+ (u\circ \bar u)(x_1, x_1) v_N(x_1-x_2)\bar \phi(x_2) u(y_1, x_2) \\
 &+ (\bar u\circ u)(x_1, x_2) v_N(x_1-x_2)\bar\phi(x_2)  u(y_1, x_1) \Big\}.
\end{align}
\end{subequations}
For $X_2, X_3, X_4$, up to normalization and symmetrization, we have that
\begin{subequations}
\begin{align}
 X_2(y_1, y_2)=&\frac{1}{2N}\int dx_1dx_2 \Big\{  \nonumber\\
 &\quad\ (u\circ c)(x_1, x_2)v_N(x_1-x_2)\bar c(y_1, x_1)c(x_2, y_2) \\ 
  &+(\bar u\circ  \bar c)(x_1, x_2)v_N(x_1-x_2)u(y_1, x_1)u(x_2, y_2) \\
      &+(u \circ \bar u)(x_1, x_2) v_N(x_1-x_2) \bar u(y_1, x_1) \bar c(x_2, y_2) \\
      &+(u \circ \bar u)(x_1, x_2) v_N(x_1-x_2)  \bar c(y_1, x_1) u(x_2, y_2)  \\
          &+(u\circ \bar u)(x_1, x_1) v_N(x_1-x_2) u(y_1, x_2)  c(x_2, y_2) \\
  &+(u \circ \bar u)(x_2, x_2) v_N(x_1-x_2) \bar u(y_1, x_1) \bar c(x_1, y_2) 
    \Big\}.
\end{align}
\end{subequations}
\begin{subequations}
\begin{align}
X_3(y_1, y_2, y_3)=&\frac{1}{\sqrt{N}}\int dx_1dx_2  \Big\{  \nonumber\\
&\quad\   u(y_1, x_1)v_N(x_1-x_2)\phi(x_2)\bar c(y_2, x_1)  c(x_2, y_3)\\
&+ u(y_1, x_1)v_N(x_1-x_2)\phi(x_2) u(y_2, x_2) c(x_1, y_3) 
 \Big\}
\end{align}
\end{subequations}
\begin{align}
X_4(y_1, y_2, y_3, y_4)=& \frac{1}{2N}\int dx_1dx_2 \Big\{  \\
&\quad\ u(y_1, x_1)u(x_2, y_2)v_N(x_1-x_2) \bar c(y_3, x_1)c(x_2, y_4)  \Big\}.\nonumber
\end{align}

\subsection{Estimates for $\mathcal{P}=\mathcal{H}_\text{red}-\mathcal{H}-X_0$} Let us split $\mathcal{P}$ as follows
\begin{align}
\mathcal{P}=\mathcal{P}_1+\mathcal{P}_2+\mathcal{P}_3 +\mathcal{P}_4
\end{align}
where each $\mathcal{P}_i$ corresponds to the homogeneous polynomial of degree $i$. 

\begin{lemma}\label{p1-lem}
Fix $k \in \nn$ (for our purposes, $k = 9$). Let $X$ be a Fock space vector that has nonzero entries only in the first $k$ sectors. Then $\mathcal{P}_1$ is a bounded operator and
\begin{align}
\llp{\mathcal{P}_1 X(t)}_\mathcal{F}\le \frac{CN^{-\frac{1}{2}}}{1+t^\frac{3}{2}} \sup_z \llp{u(\cdot, \cdot-z)}_{L^2(dx)} \llp{X}_{\mathcal{F}}.
\end{align}
In particular, it follows we have the estimate
\begin{align}
\int^t_0 d\tau\ \llp{\mathcal{P}_1 X(\tau)}_\mathcal{F} \lesssim_k N^{-\frac{1}{2}}C_0(t, N)\llp{X}_{L^\infty_\tau([0, t])\mathcal{F}}.
\end{align}
\end{lemma}

\begin{proof}
The proof of the lemma is similar to the proof of Lemma 8 in \cite{Kuz2} with the improvement coming from our Corollary \ref{main-cor}. However, there is still an
essential difference between the two proofs. Our proof avoids the usage of trace theorem when estimating $\llp{s_1(x, x+z)}_{L^2(dx)}$.

Note all the terms of $\mathcal{P}_1$ are given by \eqref{0-1terms}. By duality, it suffices to consider only the $a$ terms. Expanding \eqref{0-1terms} by $c = \delta+p$, we see that the worst terms usually come from terms with the most $\delta$s. In the case  of $\mathcal{P}_1$, we see that
the worst term comes from \eqref{0-1term1}. In fact, it is not hard to see that if we can handle \eqref{0-1term1}, then all other terms \eqref{0-1term2}-\eqref{0-1term6} can be handle in the exact same manner.

Expanding \eqref{0-1term1} we get that
\begin{subequations}
\begin{align}
\eqref{0-1term1} =  &\frac{1}{\sqrt{N}}\int dx_1dx_2 dy_1\Big\{ \nonumber \\
 &\quad\ u(x_1, x_2)v_N(x_1-x_2)\phi(x_2)\delta(y_1- x_1) \label{p1-term1}\\
   &+ u(x_1, x_2)v_N(x_1-x_2)\phi(x_2)p(y_1, x_1) \label{p1-term2}\\
&+ (\bar p\circ u)(x_1, x_2)v_N(x_1-x_2)\phi(x_2)\delta(y_1-x_1) \label{p1-term3}\\
 &+ (\bar p\circ u)(x_1, x_2)v_N(x_1-x_2)\phi(x_2)p(y_1, x_1) \label{p1-term4}
 \Big\}a_{y_1}.
\end{align}
\end{subequations}
Let $X= (0, \ldots, F(y_1, \ldots, y_n), 0, \ldots)$, then we see that
the action of \eqref{p1-term1} on $X$ yields the function, up to normalization, 
\begin{align}\label{p1-term1-X}
\frac{1}{\sqrt{N}}\int dx_1dx_2\ u(x_1, x_2) v_N(x_1-x_2)\phi(x_2) F(x_1, y_1, \ldots, y_{n-1})
\end{align}
in the $n-1$ sector. Hence we have the estimate
\begin{align*}
&\llp{\eqref{p1-term1-X}}_{L^2(dy_1, \ldots,d y_{n-1})} \\
&\le \frac{C}{\sqrt{N}}\int dx_1 dx_2\ |u(x_1, x_2)v_N(x_1-x_2)\phi(x_2)|\llp{F(x_1, \cdots)}_{L^2(dy_1\cdots dy_{n-1})}\\
&\le \frac{C}{\sqrt{N}} \llp{\phi}_{L^\infty_x} \int dx_2 |v_N(x_2)|\int dx_1\ |u(x_1, x_1-x_2)|\llp{F(x_1, \cdots)}_{L^2(dy_1\cdots dy_{n-1})}
\end{align*}
Finally, it follows that
\begin{align*}
&\int^t_0 d\tau\ \llp{\eqref{p1-term1-X}}_{L^2(dy_1, \ldots,d y_{n-1})}\\
&\le \frac{C}{\sqrt{N}}\left( \int^t_0 \frac{d\tau}{1+\tau^3}\right)^\frac{1}{2} \int dx_2\  |v_N(x_2)|\\
&\quad \times \Lp{\llp{u(x_1, x_1-x_2)}_{L^2(dx_1)}\llp{F(t)}_{L^2(dx_1dy_1\cdots dy_{n-1})}}{L^2_t}\\
&\le \frac{C}{\sqrt{N}} \sup_z \llp{u(x, x-z)}_{L^2(dtdx)}\sup_t \llp{F(t)}_{L^2(dy_1\cdots dy_n)}.
\end{align*}

For \eqref{p1-term2}, we see that the only difference in its action on $X$ is the additional composition, i.e.
\begin{align}\label{p1-term2-X}
 \frac{1}{\sqrt{N}}\int dx_1dx_2\ u(x_1, x_2) v_N(x_1-x_2)\phi(x_2) (p\circ F)(x_1, y_1, \ldots, y_{n-1}).
\end{align}
Then it follows
\begin{align*}
&\int^t_0 d\tau\ \llp{\eqref{p1-term2-X}}_{L^2(dy_1, \ldots,d y_{n-1})}\\
&\le \frac{C}{\sqrt{N}}\sup_z\Lp{\llp{u(x, x-z)}_{L^2(dx)}\llp{(p\circ F)(t)}_{L^2(dx_1dy_1\cdots dy_{n-1})}}{L^2_t}.
\end{align*}
Since we know that $p$ is a bounded operator from \eqref{p1-est}, then we have the desired result. 

Lastly, the estimates for \eqref{p1-term3} and \eqref{p1-term4} follows the same line of argument and the fact that $\llp{(u\circ p)(x, x-z)}_{L^2(dx)} \le \llp{u}_{L^2}\llp{p}_{L^2}$.
\end{proof}

Unlike $\mathcal{P}_1$, $\mathcal{P}_3$ is not a uniform in $N$ bounded operator from $\mathcal{F}$ to $\mathcal{F}$. This can be checked by applying $\mathcal{P}_3$ to $\Omega$. However, if we split 
$\mathcal{P}_3 = \mathcal{P}_3^r+\mathcal{P}_3^s$ where $\mathcal{P}_3^s:= \eqref{0-3term1} + \eqref{1-2term1}+\text{ adjoints}$ then we can show that $\mathcal{P}_s^r$ is a uniform in $N$ bounded operator from Fock space to Fock space. 

\begin{lemma} \label{p3r-lem}
Fix $k \in \nn$. Let $X$ be a Fock space vector that has nonzero entries only in the first $k$ sectors. Then $\mathcal{P}_3^r$ is a bounded operator and
\begin{align}
\llp{\mathcal{P}_3^r X(t)}_\mathcal{F}\le \frac{CN^{-\frac{1}{2}}}{1+t^\frac{3}{2}} \llp{X}_{\mathcal{F}}.
\end{align}
In particular, it follows we have the estimate
\begin{align}
\int^t_0 d\tau\ \llp{\mathcal{P}_3^r X(\tau)}_\mathcal{F} \lesssim_k N^{-\frac{1}{2}}\llp{X}_{L^\infty_\tau([0, t])\mathcal{F}}.
\end{align}
\end{lemma}

\begin{proof}
For the cubic terms, it suffices to consider \eqref{0-3term2} and \eqref{1-2term2}.  

Expanding \eqref{0-3term2} gives 
\begin{subequations}
\begin{align}
\eqref{0-3term2}=&\frac{1}{\sqrt{N}}\int dx_1dx_2 dy_1 dy_2 dy_3  \Big\{  \nonumber\\
&\quad\  \bar u(y_1, x_1)v_N(x_1-x_2)\bar\phi(x_2)\bar u(y_2, x_2) \delta(x_1-y_3) \label{p3r-term1}\\
&\quad\  \bar u(y_1, x_1)v_N(x_1-x_2)\bar\phi(x_2)\bar u(y_2, x_2) \bar p(x_1, y_3) \label{p3r-term2}\Big\}a_{y_1}a_{y_2} a_{y_3}.
\end{align}
\end{subequations}
Letting \eqref{p3r-term2} act on $F(y_1, \ldots, y_n)$, up to symmetry, yields
\begin{align}\label{p3-term2-X}
\frac{1}{\sqrt{N}} \int dz_1 dz_2 dx_1 dx_2\ \bar u(z_1, x_1) v_N(x_1-x_2)\phi(x_2) \bar u(z_2, x_2) F(x_1, z_1, z_2, \cdots).
\end{align}
By direct calculation, we have that
\begin{align*}
&\llp{\eqref{p3-term2-X}}_{L^2(dy_1\cdots dy_{n-3})}\\
&\le \frac{C}{\sqrt{N}}\llp{\phi}_{L^\infty_x} \int dz_1dz_2dx_1 dx_2\ |v_N(x_1)u(x_1+x_2, z_1)u(x_2, z_2)|\\
&\quad \times \llp{F(x_1+x_2, z_1, z_2, \cdots)}_{L^2(dy_1\cdots dy_{n-3})}\\
&\le  \frac{C}{\sqrt{N}}\llp{\phi}_{L^\infty_x} \int dx_1 dx_2\ |v_N(x_1)\llp{u(x_1+x_2, \cdot)}_{L^2_y}\llp{u(x_2,\cdot)}_{L^2_y}\\
&\quad \times \llp{F(x_1+x_2, \cdots)}_{L^2(dy_1\cdots dy_{n-1})}\\
&\le  \frac{C}{\sqrt{N}}\llp{\phi}_{L^\infty_x}\llp{v}_{L^1_x}\llp{u}_{L^4_xL^2_y}^2\llp{F}_{L^2(dy_1\cdots dy_n)}
\end{align*}
which yields the desired result. \eqref{p3r-term2} follows immediately from the above argument and the fact that $p$ is a uniform in $N$ bounded operator. 

Next, expanding  \eqref{0-3term2} gives 
\begin{subequations}
\begin{align}
&\frac{1}{\sqrt{N}}\int dx_1dx_2 dy_1 dy_2 dy_3  \Big\{  \nonumber\\
&\quad\ \delta(y_1-x_1)v_N(x_1-x_2)\phi(x_2)\bar u(y_2, x_2) \delta(x_1- y_3)\label{p3r-term3} \\
&+ \delta(y_1- x_1)v_N(x_1-x_2)\phi(x_2)\bar u(y_2, x_2) \bar p(x_1, y_3) \label{p3r-term4}\\
&+ p(y_1, x_1)v_N(x_1-x_2)\phi(x_2)\bar u(y_2, x_2) \delta(x_1- y_3) \label{p3r-term5}\\
&+ p(y_1, x_1)v_N(x_1-x_2)\phi(x_2)\bar u(y_2, x_2) \bar p(x_1, y_3) \Big\} a^\ast_1 a_2a_3\label{p3r-term6}
\end{align}
\end{subequations}
For \eqref{p3r-term3}, we have to handle the term
\begin{align}\label{p3r-term3-X}
\frac{1}{\sqrt{N}}\int dx dz\ v_N(y_1-x)\phi(x)\bar u(x, z)F(z, y_1, y_2, \cdots ).
\end{align}
By direct calculation, we see that
\begin{align*}
&\llp{\eqref{p3r-term3-X}}_{L^2(dy_1\cdots dy_n)} \\
&\le \frac{1}{\sqrt{N}}\Lp{ \int dxdz\ |v_N(y_1-x)\phi(x)\bar u(x, z)| \llp{F(z, y_1, \cdots)}_{L^2(dy_2\cdots dy_n)}}{L^2(dy_1)}\\
&\le \frac{\llp{\phi}_{L^\infty_x}}{\sqrt{N}}\Lp{ \int dxdz\ |v_N(x)\bar u(y_1-x, z)| \llp{F(z, y_1, \cdots)}_{L^2(dy_2\cdots dy_n)}}{L^2(dy_1)}\\
&\le \frac{\llp{\phi}_{L^\infty_x}}{\sqrt{N}}\Lp{ \int dx\ |v_N(x)| \llp{\bar u(y_1-x, \cdot)}_{L^2_y} \llp{F(\cdot, y_1, \cdots)}_{L^2(dzdy_2\cdots dy_n)}}{L^2(dy_1)}\\
&\le \frac{C}{\sqrt{N}}\llp{\phi}_{L^\infty_x}\llp{u}_{L^\infty_xL^2_y}\llp{v}_{L^1_x}\llp{F}_{L^2(dy_1\cdots dy_n)}.
\end{align*}
Estimates for \eqref{p3r-term4}-\eqref{p3r-term6} follow from the above argument and the boundedness of $p$. 
\end{proof}

\begin{lemma}\label{p3s-lem}
Fix $k \in \nn$. Let $X$ be a Fock space vector that has nonzero entries only in the first $k$ sectors. Then we have the estimates
\begin{align}
\llp{\mathcal{P}_3^s X(t)}_{\mathcal{F}}\le&\ \frac{CN^{\frac{3\beta-1}{2}}}{1+t^\frac{3}{2}} \llp{X}_{\mathcal{F}}, \\
\llp{\mathcal{P}_3^s X}_{S'}\le&\ CN^{\frac{\beta-1}{2}}\llp{X}_{S}.
\end{align}
\end{lemma}

\begin{proof}[Sketch of Proof]
The proof is essentially the same as the proof of Proposition 10.1 in \cite{GM3}. The only difference between the proofs is the fact that we put
 $\phi \in L^\infty_x$, instead of $L^2$, then apply the $L^\infty$  decay estimate.
\end{proof}

For the quartic terms $\mathcal{P}_4$, we have the following lemma.

\begin{lemma}\label{p4-lem}
Fix $k \in \nn$. Let $X$ be a Fock space vector that has nonzero entries only in the first $k$ sectors. Then $\mathcal{P}_4$ is a bounded operator and
\begin{align}
\llp{\mathcal{P}_4 X}_\mathcal{F}\le CN^{\frac{3\beta-2}{2}} \llp{X}_{\mathcal{F}}.
\end{align}
In particular, it follows we have the estimate
\begin{align}
\int^t_0 d\tau\ \llp{\mathcal{P}_4X(\tau)}_\mathcal{F} \lesssim_k N^{\frac{3\beta-2}{2}}t\llp{X}_{L^\infty_\tau([0, t])\mathcal{F}}.
\end{align}
\end{lemma}
\begin{proof}[Sketch Proof]
This is Proposition 10.4 in \cite{GM3}. 
\end{proof}

Recall the quadratic terms are given by 
\begin{equation}
\begin{aligned}
\mathcal{P}_2 =&\ \int dxdy\ (-\lapl_x\delta(x-y)+L(t, x, y))a_x^\ast a_y + \eqref{sa-1-1terms}+\eqref{1-1terms} + \eqref{0-2terms}.
\end{aligned}
\end{equation}
We will split $\mathcal{P}_2$ into two groups: The first group comprise of $\mathcal{L}:=\int dxdy\ (-\lapl_x\delta(x-y)+L(t, x, y))a_x^\ast a_y$. Then we split 
\eqref{sa-1-1terms}--\eqref{0-2terms} into a ``regular" part and a ``singular" part, denoted by $\mathcal{P}_2^r$ and $\mathcal{P}_2^s$, where
\begin{align}
\mathcal{P}_2^s = \eqref{1-1term1}+\eqref{0-2term1}.
\end{align}

\begin{lemma}\label{p2r-lem}
Fix $k \in \nn$. Let $X$ be a Fock space vector that has nonzero entries only in the first $k$ sectors. Then $\mathcal{P}_2^r$ is a bounded operator 
and
\begin{align}
\llp{\mathcal{P}_2^r X(t)}_\mathcal{F}\le CN^{-1} \sup_z \llp{s_2(\cdot, \cdot-z)}_{L^2(dx)} \llp{X}_{\mathcal{F}}.
\end{align}
Moreover, we have that
\begin{align}
\int^t_0 d\tau\ \llp{\mathcal{P}_2^r X(\tau)}_\mathcal{F}\le N^{-1} C_0(t, N)\llp{X}_{L^\infty_\tau([0, t])\mathcal{F}}.
\end{align}
\end{lemma}

\begin{proof}
For the \eqref{sa-1-1terms} group, it suffices to just consider \eqref{sa-1-1term1} and \eqref{sa-1-1term6}. 

For \eqref{sa-1-1term1}, it suffices to just consider the
\begin{align}\label{sa-1-1term1-X}
&\frac{1}{2N}\int dx_1dx_2 dy_1 dy_2 \Big\{  \nonumber\\
 &+  (u\circ \bar u)(x_1, x_1)v_N(x_1-x_2) \delta(y_1-x_2) \delta(x_2-y_2)a^\ast_2 a_1 \Big\}.
\end{align}
Let \eqref{sa-1-1term1-X} act on $F(y_1, \ldots, y_n)$ yields
\begin{align}\label{sa-1-1term1-F} 
&\frac{1}{2N}\int dx_1\  (u\circ \bar u)(x_1, x_1)v_N(x_1-y_1)F(y_1, y_2, \ldots, y_n).
\end{align}
Then it follows
\begin{align*}
&\llp{\eqref{sa-1-1term1-F}}_{L^2(dy_1\cdots dy_n)} \\
&\le \frac{C}{N}\Lp{\int dx\  (u\circ \bar u)(x, x)v_N(x-y_1)}{L^2(dy_1)}\llp{F}_{L^2(dy_1\cdots dy_n)}\\
&\le \frac{C}{N}\int dx\ |v_N(x)| \llp{(u\circ u)(x+y_1, x+y_1)}_{L^2(dy_1)}\llp{F}_{L^2(dy_1\cdots dy_n)}\\
&\le \frac{C}{N}\llp{v}_{L^1_x} \llp{u}_{L^4_xL^2}^2\llp{F}_{L^2(dy_1\cdots dy_n)}.
\end{align*}

For \eqref{sa-1-1term6}, we see its action on $F$ yields
\begin{align}\label{sa-1-1term6-X}
 &\frac{1}{2N}\int dx_1dx_2 dz \Big\{  \nonumber\\
&(u\circ \bar u)(x_2, x_1)v_N(x_1-x_2)u(y_1, x_1) \bar u(x_2, z)F(z, y_2,\ldots, y_n)\Big\}.
\end{align}
Then we have that
\begin{align*}
&\llp{\eqref{sa-1-1term6-X}}_{L^2(dy_1\cdots dy_n)}\\
&\le \frac{C}{N}\int dx_1dx_2 dz \Big\{  \\
&\quad\ |(u\circ \bar u)(x_2, x_1)v_N(x_1-x_2)\bar u(x_2, z)|\llp{u(\cdot, x_1)}_{L^2_x} \llp{F(z, \cdots)}_{L^2(dy_2\cdots dy_n)}\Big\}\\
&\le \frac{C}{N}\llp{u}_{L^\infty_xL^2_y}\int dx_1dx_2 dz \Big\{  \\
&\quad\ |v_N(x_1)(u\circ \bar u)(x_2, x_1+x_2)\bar u(x_2, z)| \llp{F(z, \cdots)}_{L^2(dy_2\cdots dy_n)}\Big\}\\
&\le \frac{C}{N}\llp{u}_{L^\infty_xL^2_y}^2\int dx_1dx_2\ |v_N(x_1)(u\circ \bar u)(x_2, x_1+x_2)| \llp{F}_{L^2(dy_1dy_2\cdots dy_n)}\\
&\le \frac{C}{N}\llp{v}_{L^1_x}\llp{u}_{L^\infty_xL^2_y}^2\llp{u}^2_{L^2_{x, y}} \llp{F}_{L^2(dy_1dy_2\cdots dy_n)}.
\end{align*}

Now, for the \eqref{1-1terms} and \eqref{0-2terms} group, it suffices to consider \eqref{0-2term2} since the other terms are similar to terms in group \eqref{sa-1-1terms}.

Observe the action of \eqref{0-2term2} yields
\begin{align}\label{0-2term2-X}
 &\frac{1}{4N}\int dx_1dx_2 dz_1 dz_2 \Big\{  \nonumber\\
  &s_2(x_1, x_2)v_N(x_1-x_2)\bar u(z_1, x_1)\bar u(x_2, z_2)F(z_1, z_2, y_1, \ldots, y_{n-2})\Big\}.
\end{align}
Then we have that
\begin{align*}
&\llp{\eqref{0-2term2-X}}_{L^2(dy_1\cdots dy_{n-2})}\\
&\le \frac{C}{N}\int dx_1dx_2 dz_1 dz_2 \Big\{  \nonumber\\
  &\quad\ |s_2(x_1, x_2)v_N(x_1-x_2)\bar u(z_1, x_1)\bar u(x_2, z_2)|\llp{F(z_1, z_2, \cdots)}_{L^2(y_1, \ldots, y_{n-2})}\Big\}\\
  &\le \frac{C}{N}\llp{u}_{L^\infty_xL^2_y}\int dx_1dx_2\ |s_2(x_1, x_2)v_N(x_1-x_2)|\llp{u(x_2, \cdot)}_{L^2_y} \llp{F}_{L^2(y_1, \ldots, y_{n})}\\
   &\le \frac{C}{N}\llp{u}_{L^2_{x, y}} \llp{u}_{L^\infty_xL^2_y}\int dz\ |v_N(z)|\int dx\ |s_2(z+x, x)|^2\llp{F}_{L^2(y_1, \ldots, y_{n})}.
\end{align*}
Finally, integrate with respect to time and apply \eqref{s2-collapsing-est} yields the desired result. 
\end{proof}

\begin{lemma}\label{p2s-lem}
Fix $k \in \nn$. Let $X$ be a Fock space vector that has nonzero entries only in the first $k$ sectors. Then we have the following estimates
\begin{subequations}
\begin{align}
\int^t_0 d\tau\ \llp{\mathcal{P}_2^s X(\tau)}_\mathcal{F}\le&\ CN^{\frac{3\beta-2}{2}} \sqrt{t}C_0(t, N) \llp{X}_{L^\infty_\tau ([0, t])\mathcal{F}}\\
\llp{\mathcal{P}_2^sX}_{S'} \le&\ CN^\frac{\beta-2}{2}C_0(t, N)\llp{X}_S.
\end{align}
\end{subequations}
\end{lemma}

\begin{proof}
It suffices to consider \eqref{0-2term1}. Again, it also suffices to just consider the $\delta$ terms. 
The action of \eqref{0-2term1} on $F$ yields
\begin{align}\label{0-2term1-X}
\frac{1}{4N}\int dx_1dx_2\ s_2(x_1, x_2)v_N(x_1-x_2)F(x_1, x_2, y_1, \ldots, y_{n-2}).
\end{align}
Then we have that
\begin{align*}
&\int^t_0 d\tau\ \llp{\eqref{0-2term1-X}}_{L^2(dy_1\cdots dy_{n-2})}\\
& \le \frac{C}{N} \int^t_0 d\tau\  \llp{s_2(x_1, x_2)v_N(x_1-x_2)}_{L^2(dx_1dx_2)}\llp{F}_{L^2(dy_1\cdots dy_{n})}\\
&\le \frac{C}{N}\llp{v_N}_{L^2_x} \sqrt{t} \llp{s_2(x, x+z)}_{L^2_\tau([0, t])L^2_x}\sup_\tau \llp{F(\tau)}_{L^2(dy_1\cdots dy_{n})}.
\end{align*}

For the Strichartz estimates, let us begin by writing $F(x_1, x_2, \cdots) = G(x_1-x_2, x_1+x_2, \cdots)$, then we see that
\begin{align*}
&\int^t_0 d\tau\ \llp{\eqref{0-2term1-X}}_{L^2(dy_1\cdots dy_{n-2})}\\
&\le \frac{C}{N}\int^t_0 d\tau \int dx_1 dx_2\  |s_2(x_1, x_2) v_N(x_1-x_2)| \llp{G(x_1-x_2, x_1+x_2, \cdot)}_{L^2}\\
&\le \frac{C}{N}\int^t_0 d\tau \int dx_2 dx_1\  |s_2(x_1, x_1-x_2) v_N(x_2)| \llp{G(x_2, 2x_1-x_2, \cdot)}_{L^2}\\
&\le \frac{C}{N}\int^t_0 d\tau \int dx_2\ |v_N(x_2)| \llp{s_2(x, x-x_2)}_{L^2_x}\llp{G(x_2,  \cdot)}_{L^2L^2}\\
&\le \frac{C}{N}\int^t_0 d\tau\ \llp{v_N}_{L^\frac{6}{5}_x} \llp{s_2(x, x-x_2)}_{L^2_x}\llp{G}_{L^6_{y_1}L^2_{y_2}L^2}\\
&\le \frac{C}{N} \llp{v_N}_{L^\frac{6}{5}_x} \sup_z\llp{s_2(x, x-z)}_{L^2_tL^2_x}\llp{G}_{L^2_tL^6_{y_1}L^2_{y_2}L^2}\\
& = \frac{C}{N} \llp{v_N}_{L^\frac{6}{5}_x} \sup_z\llp{s_2(x, x-z)}_{L^2_tL^2_x}\llp{F}_{L^2_tL^6(d(x_1-x_2))L^2(d(x_1+x_2))L^2}
\end{align*}
This completes the proof.
\end{proof}

\begin{lemma}\label{L-lem}
Fix $k \in \nn$. Let $X$ be a Fock space vector that has nonzero entries only in the first $k$ sectors. Then $\mathcal{L}$ is a uniform in $N$ bounded operator
and we have the estimate
\begin{align}
\llp{\mathcal{L}X}_\mathcal{F} \leq \frac{C}{1+t^3}\llp{X}_\mathcal{F}.
\end{align}
\end{lemma}
\begin{proof}
See Lemma 11 in \cite{Kuz2} for the proof of the lemma.
\end{proof}

\theendnotes
\providecommand{\bysame}{\leavevmode\hbox to3em{\hrulefill}\thinspace}
\providecommand{\MR}{\relax\ifhmode\unskip\space\fi MR }
\providecommand{\MRhref}[2]{%
  \href{http://www.ams.org/mathscinet-getitem?mr=#1}{#2}
}
\providecommand{\href}[2]{#2}

\end{document}